\documentclass[11pt,table]{article}

\usepackage[utf8]{inputenc}
\usepackage[T1]{fontenc}
\usepackage[english]{babel}
\usepackage{amsmath}
\usepackage{amssymb}
\usepackage{tabularx}
\usepackage{multicol}
\usepackage{graphicx}
\usepackage{authblk}
\usepackage{layout}
\usepackage{dsfont}
\usepackage{bbold}
\usepackage[papersize={8.5in,11in}, top=1.0in, bottom=1.3in, left=1.0in, right=1.0in]{geometry}
\usepackage{amsthm}
\usepackage{mathrsfs}
\usepackage{tikz}
\usepackage{tikz-cd}
\usetikzlibrary{calc, positioning, fit, shapes.misc}
\usepackage{hyperref}
\usepackage{complexity}
\usepackage{enumitem}
\usepackage{fontawesome}
\usepackage{soul}
\usepackage{float}
\usepackage[margin=2cm]{caption}
\newtheorem{theorem}{Theorem}[section]

\newtheorem{claim}[theorem]{Claim}
\newtheorem{lemma}[theorem]{Lemma}
\theoremstyle{definition}
\newtheorem{definition}[theorem]{Definition}
\newtheorem{remark}[theorem]{Remark}

\newcommand{\Ga}{\mathcal{P}}
\newcommand{\strat}{\Sigma}
\newcommand{\sequence}{\mathcal{S}}

\definecolor{mygray}{rgb}{0.86, 0.86, 0.86}

\title{\LARGE 4-uniform Maker-Breaker and Maker-Maker games are PSPACE-complete}
\author{\Large Florian Galliot}
\affil{\large Aix-Marseille Université, CNRS, I2M, UMR 7373, 13453 Marseille, France}
\date{}

\begin{document}

\maketitle

\begin{abstract}
    We study two positional games played on hypergraphs, whose edges may be interpreted as winning sets. Two players take turns picking a previously unpicked vertex of the hypergraph. We say a player fills an edge if that player has picked all the vertices of that edge. In the Maker-Maker convention, whoever first fills an edge wins, or we get a draw if no edge is filled. In the Maker-Breaker convention, the first player aims at filling an edge while the second player aims at preventing the first player from filling an edge. Our main result is that, for both games, deciding whether the first player has a winning strategy is a {\sf PSPACE}-complete problem even when restricted to 4-uniform hypergraphs (of bounded maximum degree). For the Maker-Maker convention, this improves on the known {\sf PSPACE}-completeness result for hypergraphs of rank 4. For the Maker-Breaker convention, this improves on the known {\sf PSPACE}-completeness result for 5-uniform hypergraphs, and closes the complexity gap since the problem for hypergraphs of rank 3 is known to be solvable in polynomial time. As a corollary of our construction, we actually get a stronger result: deciding whether the first player has a winning strategy for the vertex-$C_4$-game played on arbitrary graphs, where the winning sets are the vertex sets of 4-cycles, is a {\sf PSPACE}-complete problem for both conventions.
\end{abstract}

\section{Introduction}\label{section1}\strut
\indent\textbf{Positional games.} {\em Positional games} have been introduced by Hales and Jewett \cite{Hales1963}, and later popularized by Erd\H{o}s and Selfridge~\cite{erdos}. The game board is a hypergraph $H$, with vertex set $V(H)$ and edge set $E(H)\subseteq \mathcal{P}(V(H))$. Two players take turns picking a previously unpicked vertex, and the result of the game is defined by one of several possible {\em conventions}. The study of positional games mainly consists, for various conventions and classes of hypergaphs, in trying to determine the {\em outcome} of the game i.e. the result of the game when both players play optimally, or working out the algorithmic complexity of computing the outcome.

\medskip

\indent\textbf{Maker-Maker games.} The {\em Maker-Maker} convention was the first one to be introduced in all generality, by Hales and Jewett in 1963 \cite{Hales1963}: the edges are the winning sets, meaning that whoever first {\em fills} any edge (i.e. picks all the vertices of some edge) wins the game. If no one has filled any edge by the time all vertices are picked, we get a draw. A well-known {\em strategy-stealing} argument \cite{Hales1963} shows that there are only two possible outcomes: a first player win or a draw. Indeed, if the second player SP had a winning strategy, then the first player FP could play an arbitrary first move, proceeding to act as second player and ``steal'' that strategy to win herself. The most famous example of a Maker-Maker game is tic-tac-toe, which can be generalized into the {\em $k$-in-a-row} game on grids of any size \cite{Bec08}. For instance, the $k$-in-a-row game on an infinite grid is known to be a first player win for $k \leq 4$ but a draw for $k \geq 8$ \cite{Zetters}, while the cases $k \in \{5,6,7\}$ are important open problems.

In the Maker-Maker convention, both players must manage offense and defense at the same time, by trying to fill an edge while also preventing the opponent from doing so first. This creates counterintuitive phenomena such as the {\em extra edge paradox} \cite{Bec08}, where adding an edge might turn a first player win into a draw. As such, the Maker-Maker convention is notoriously difficult to handle.

\medskip

\indent\textbf{Maker-Breaker games.} The {\em Maker-Breaker} convention was introduced by Chv\'atal and Erd\H{o}s in 1978 \cite{CE78}. The players are called Maker and Breaker: Maker goes first, aiming at filling an edge, while Breaker's goal is to prevent Maker from filling an edge. Since draws are impossible, there are only two possible outcomes: a Maker win or a Breaker win. Also note that, if Breaker has a winning strategy for the Maker-Breaker game on $H$, then SP can use the same strategy as a drawing strategy for the Maker-Maker game on $H$: as such, compared to the Maker-Maker convention, the Maker-Breaker convention makes things easier for the first player. The board game Hex is the best known example of a Maker-Breaker game \cite{Gardner_Hex}.

This convention is by far the most studied, as it possesses some convenient additional properties compared to the Maker-Maker convention, such as monotonicity of the outcome under taking subhypergraphs. 
Aside from that, the Maker-Breaker convention is also natural because it can be interpreted as a two-player version of the Boolean satisfiability problem (\textsc{SAT}) on positive CNF formulas. Indeed, identifying the hypergraph as a CNF formula where the vertices are variables and the edges are clauses with all positive literals, and Maker's (resp. Breaker's) picks as set to {\sf False} (resp. {\sf True}), we see that Breaker aims at satisfying the formula while Maker aims at falsifying it. This standpoint allows for generalizing the Maker-Breaker game to CNF formulas that need not be monotone: a move then consists in choosing not only a variable but also a truth value for it (indeed, in the absence of monotonicity, this choice is no longer obvious) \cite{UnorderedQBF1,UnorderedQBF2}. This generalized game is a close relative of the quantified Boolean formula problem (\textsc{QBF}), which is a canonical {\PSPACE}-complete problem \cite{QBF}. As such, we may call it \textsc{UnorderedQBF}. The sole difference is that \textsc{QBF} has preordered variables, so that the players only choose the truth values.

\medskip

\indent\textbf{Positional games on graphs.} Historically, Maker-Breaker games were first introduced on graphs, with the players picking edges of a complete graph. Maker's goal is to achieve some substructure, for example a spanning tree or a Hamiltonian cycle \cite{CE78}. In this setting, the size discrepancy between the complete graph and the objective subgraph means Breaker is at a disadvantage. One way to rebalance the game is to give Breaker a bias $b$, meaning he picks edges in bunches of size $b$, and the main object of study then becomes the threshold value of $b$ over which Breaker has a winning strategy (see \cite{Milos} for a survey on such games). An alternative way is to play on arbitrary graphs instead of complete graphs. A famous example is the Shannon switching game \cite{Gardner_Shannon,Lehman}, in which Maker tries to connect two given terminals. There is also the {\em edge-$H$-game}, in which the winning sets are the edge sets of subgraphs isomorphic to some given graph $H$ \cite{Hgame}.

Positional games played on vertex sets of graphs also exist in the literature. It should be noted that one such game is Hex, which has been generalized to larger boards \cite{Reisch} as a special case of the version of the Shannon switching game played on vertices \cite{EvenTarjan}. The Maker-Breaker domination game, in which the winning sets are the dominating sets, has been studied extensively since its introduction in 2020 \cite{domination1} and is a rare occurrence of a positional game on graphs that has also been studied in the Maker-Maker convention \cite{domination2}. Another example is the {\em vertex-$H$-game}, in which the winning sets are the vertex sets of subgraphs isomorphic to some given graph $H$  \cite{vertexHgame}.

\medskip

\indent\textbf{Algorithmic complexity.} The algorithmic complexity of positional games is studied on finite hypergraphs and is usually discussed depending on the size of the edges. A hypergraph is of {\em rank} $k$ if its largest edge has size $k$, or {\em $k$-uniform} if all its edges are of size $k$. As both Maker-Maker and Maker-Breaker conventions only have two possible outcomes, they naturally translate into a decision problem which asks whether the first player (FP or Maker) has a winning strategy on an input hypergraph $H$.

The Maker-Breaker problem has first been shown {\sf PSPACE}-complete even when restricted to hypergraphs of rank 11 \cite{Schaefer}, and much later to 6-uniform hypergraphs \cite{RW21}, both times via a reduction from \textsc{3-QBF} i.e. \textsc{QBF} with clauses of size 3. The {\sf PSPACE}-completeness result has very recently been improved to 5-uniform hypergraphs \cite{Finn} via a reduction from a restricted version of \textsc{GeneralizedGeography} \cite{Geography}. On the other hand, the problem is tractable for hypergraphs of rank 3 \cite{Florian_MB3}, leaving the 4-uniform case as the only open problem prior to the present paper. Actually, no hardness result exists even for the more general problem of \textsc{UnorderedQBF} on nonmonotone 4-CNF formulas. When it comes to Maker-Breaker games on graphs, the edge-$H$-game has been shown {\sf PSPACE}-complete for several graphs $H$ (including a tree of order 91) but tractable for some others \cite{Hgame}.

As for the Maker-Maker convention, using an argument which reduces $k$-uniform Maker-Breaker games to $(k+1)$-uniform Maker-Maker games \cite{Byskov}, the aforementioned results on the Maker-Breaker convention imply that the Maker-Maker problem is {\sf PSPACE}-complete on 6-uniform hypergraphs. Very recently, the Maker-Maker problem has been shown {\sf PSPACE}-complete on hypergraphs of rank 4 \cite{FlorianJonas_MM4}. On the other end of the spectrum, the Maker-Maker problem on hypergraphs of rank 2 is trivially tractable, but hypergraphs of rank 3 remain an open case.

\medskip

\indent\textbf{Our contribution.} We close the complexity gap for the Maker-Breaker convention, by showing that the Maker-Breaker problem is {\sf PSPACE}-complete for 4-uniform hypergraphs, even with bounded maximum degree. For this, we use the same fragment of \textsc{GeneralizedGeography} that was used by Koepke in \cite{Finn}, employing the same strategy of proof but adapting the gadgets of rank 5 to make them 4-uniform. This result is proved in Section \ref{section3}, after some definitions and preliminary lemmas in Section \ref{section2} on the three games at hand (Maker-Breaker, Maker-Maker, and Geography).
In Section \ref{section4}, we show that the Maker-Breaker construction also works in the Maker-Maker convention, improving the {\sf PSPACE}-completeness result from rank 4 previously to 4-uniform. We then show in Section \ref{section5} that the same construction gives a proof of {\sf PSPACE}-completeness for the vertex-$C_4$-game, which constitutes the first result on the algorithmic complexity of the vertex-$H$-game for both conventions. 
Finally, Section \ref{section6} concludes this paper and provides leads for future work.

\section{Preliminaries}\label{section2}\strut
\indent All hypergraphs will be assumed to be finite. To avoid confusion from the fact that the proofs feature hypergraphs and digraphs at the same time, we use the ``vertex/edge'' terminology for hypergraphs and the ``node/arc'' terminology for digraphs. We reserve the letters $p,q,x,y,z$ for naming vertices and the letters $t,s,u,v,w$ for naming nodes. 

\subsection{Positional games}\strut
\indent Let us first explain how we update the game throughout the players' moves. An important characteristic of the Maker-Breaker convention, in contrast with the Maker-Maker convention, is that intermediate positions of the game can be seen as starting positions. Indeed, consider an intermediate position on a hypergraph $H$, where we denote by $M$ and $B$ the set of vertices picked by Maker and Breaker respectively so far: then, it is as if a fresh Maker-Breaker game is about to start on the {\em updated hypergraph} $H'$ defined by $V(H')=V(H) \setminus (M \cup B)$ and $E(H')=\{e \setminus M \mid e \in E(H) \text{ such that } e \cap B= \varnothing\}$. Indeed, any edge $e$ such that $e \cap B \neq \varnothing$ may be removed from the game as Maker cannot fill it anymore (we say Breaker has {\em killed} $e$), while any edge $e$ such that $e \cap B = \varnothing$ now behaves like the edge $e \setminus M$ in the sense that Maker will finish filling $e$ if and only if she picks all vertices in $e \setminus M$. From this point of view, Maker winning the game is equivalent to the updated hypergraph containing the edge $\varnothing$. An illustration is given in Figure \ref{fig:updateMB}. Using this way of updating the game, even though we always assume that Maker goes first, the case where Breaker goes first actually is polynomially equivalent since we can ``start the game on the second move'' by iterating over all possibilities for the first move.

\begin{figure}[h]
    \centering
    \includegraphics[scale=.5]{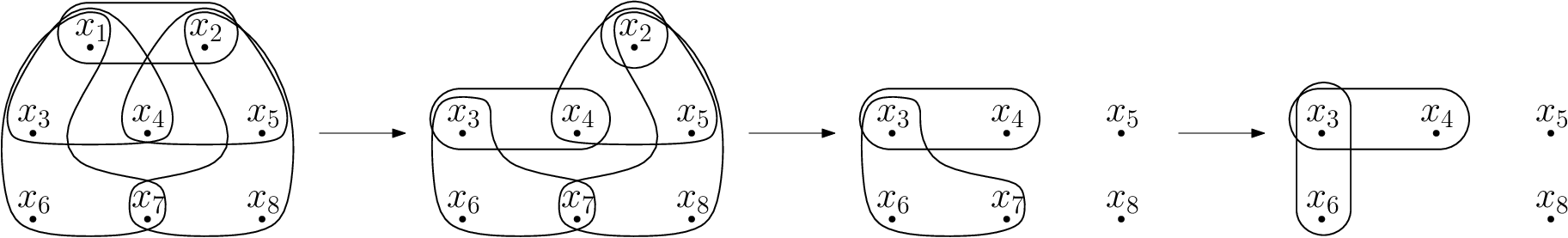}
    \caption{Updating the Maker-Breaker game on a hypergraph $H$ defined by $V(H)=\{x_1,x_2,x_3,x_4,x_5,x_6,x_7,x_8\}$ and $E(H)=\{\{x_1,x_2\},\{x_1,x_3,x_4\},\{x_2,x_4,x_5\},\{x_1,x_3,x_6,x_7\},\{x_2,x_5,x_7,x_8\}\}$. Maker picks $x_1$, then Breaker picks $x_2$, then Maker picks $x_7$.}\label{fig:updateMB}
\end{figure}

Things are different in the Maker-Maker convention. A player who picks a vertex $x$ makes progress in all edges containing $x$, but also kills those edges from the opponent's point of view. Therefore, following the idea introduced in \cite{FlorianJonas_Eurocomb,FlorianJonas_Arxiv}, we see intermediate positions of the game as having two sets of edges, red and blue, corresponding to the winning sets of FP and SP respectively. At the start of the game played on a hypergraph $H$, both sets of edges are equal to $E(H)$. After any number of moves, where we denote by $F$ and $S$ the set of vertices picked by FP and SP respectively so far, and assuming the game is not over yet, the set of {\em updated red edges} is defined as $\{e \setminus F \mid e \in E(H) \text{ such that } e \cap S= \varnothing\}$ while the set of {\em updated blue edges} is defined as $\{e \setminus S \mid e \in E(H) \text{ such that } e \cap F= \varnothing\}$. FP (resp. SP) wins the game if the edge $\varnothing$ first appears as an updated red edge (resp. as an updated blue edge). An illustration is given in Figure \ref{fig:updateMM}.

\begin{figure}[h]
    \centering
    \includegraphics[scale=.5]{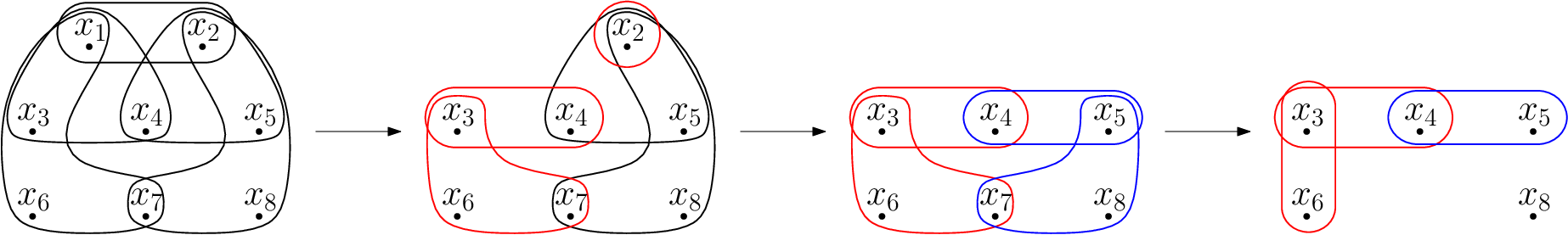}
    \caption{Updating the Maker-Maker game on the same hypergraph as in Figure \ref{fig:updateMB}, with the same moves: FP picks $x_1$, then SP picks $x_2$, then FP picks $x_7$. Updated edges that exist in both red and blue are represented in black.}\label{fig:updateMM}
\end{figure}

Let us state a couple of simple strategic principles which we will use extensively in our proofs. A {\em pairing} of a hypergraph $H$ is a set $\Pi$ of pairwise disjoint pairs of vertices of $H$ such that, for all $e \in E(H)$, there exists $\pi \in \Pi$ such that $\pi \subseteq e$. We say that a pairing $\Pi$ of $H$ {\em uses} a vertex $x \in V(H)$ if $x \in \bigcup_{\pi \in \Pi}\pi$.

\begin{lemma}[Pairing strategy {[folklore]}]\label{lem:pairing}
    Let $H$ be a hypergraph. If $H$ admits a pairing $\Pi$, then Breaker has a winning strategy for the Maker-Breaker game on $H$, and SP has a drawing strategy for the Maker-Maker game on $H$.
\end{lemma}

\begin{proof}
    Let us show the first result, as it implies the second. Anytime Maker picks some $x \in \{x,y\} \in \Pi$ with $y$ unpicked, Breaker answers by picking $y$, otherwise Breaker picks an arbitrary vertex. This strategy ensures that Breaker ends up picking at least one vertex in each pair of the pairing $\Pi$, thus preventing Maker from filling an edge.
\end{proof}

The notion of ``greedy move'' is defined in \cite[Lemma 3.8]{FlorianJonas_Arxiv} as a tool to show the optimality of some moves in hardness reductions. The idea is that, if some move $x$ forces an answer $y$ that is dominated by $x$ in the sense that every edge containing $y$ also contains $x$, then picking $x$ is optimal. 
We state a version of this result for both conventions, with essentially the same proof.

\begin{definition}[Greedy move for Maker-Breaker]\label{def:greedy}
    Let $H$ be a hypergraph, let $x \in V(H)$, and let $k \geq 1$ be an integer. Suppose that, for all $y \in V(H) \setminus \{x\}$, one of the following holds:
\begin{enumerate}[label={{\em (\roman*)}},noitemsep,nolistsep]
	\item If Maker picks $x$ and Breaker picks $y$, then Maker can guarantee to win by filling an edge in at most $k$ moves (i.e. $k$ moves by Maker, not counting $x$).
	\item For all $e \in E(H)$: $y \in e \implies x \in e$.
\end{enumerate}    
    We then say that $x$ is a {\em $k$-greedy move} (or simply a {\em greedy move}) in $H$.
\end{definition}

\begin{lemma}\label{lem:greedy}
    A greedy move is always an optimal move for Maker in the Maker-Breaker convention.
\end{lemma}

\begin{proof}
    Let $y \in V(H) \setminus \{x\}$, and let $H'$ be the updated hypergraph obtained after Maker picks $x$ and Breaker picks $y$. We must show that, if Maker has a winning strategy on $H$, then she also has a winning strategy on $H'$. We show the contrapositive: we suppose that Breaker has a winning strategy $\strat$ on $H'$, and we show that Breaker also has a winning strategy on $H$. Note that, since Breaker has a winning strategy on $H'$, item {\em (i)} does not hold, so item {\em (ii)} does i.e. $y \in e \implies x \in e$ for all $e \in E(H)$.
    
    Playing on $H$, Breaker simply follows the strategy $\strat$, except if Maker has just picked $x$ or $y$ in which case Breaker picks the other (or, if $x$ and $y$ are the only unpicked vertices, then Breaker picks one arbitrarily). Let us show that Breaker has indeed killed every edge $e \in E(H)$ in the process. If $e$ contains both $x$ and $y$, then Breaker has killed $e$ since he has picked $x$ or $y$. If $e$ contains neither $x$ nor $y$, then $e \in E(H')$, so Breaker has killed $e$ since $\strat$ is a winning strategy on $H'$. If $e$ contains $x$ but not $y$, then $e \setminus \{x\} \in E(H')$, so Breaker has killed $e$ (actually, he has killed $e \setminus \{x\}$) since $\strat$ is a winning strategy on $H'$. The case where $e$ contains $y$ but not $x$ is impossible since $y \in e \implies x \in e$, so the proof is over.
\end{proof}

\begin{definition}[Greedy move for Maker-Maker]\label{def:greedy2}
    Consider an intermediate position $\Ga$ of the Maker-Maker game, with FP next to play. Let $x$ be an unpicked vertex in $\Ga$, and let $k \geq 1$ be an integer. Suppose that, for every unpicked vertex $y \neq x$, one of the following holds:
\begin{enumerate}[label={{\em (\roman*)}},noitemsep,nolistsep]
	\item If FP picks $x$ and SP picks $y$, then FP can guarantee to win by filling an edge in at most $k$ moves (i.e. $k$ moves by FP, not counting $x$).
	\item For every edge in $\Ga$, be it red or blue: $y \in e \implies x \in e$.
\end{enumerate}
    We then say that $x$ is a {\em $k$-greedy move} (or simply a {\em greedy move}) in $\Ga$.
\end{definition}

\begin{lemma}\label{lem:greedy2}
    A greedy move is always an optimal move for FP in the Maker-Maker convention.
\end{lemma}

\begin{proof}
    Let $y \neq x$ be an unpicked vertex in $\Ga$, and let $\Ga'$ be the intermediate position obtained from $\Ga$ after FP picks $x$ and SP picks $y$. We must show that, if FP has a winning (resp. non-losing) strategy on $\Ga'$, then she also has a winning (resp. non-losing) strategy on $\Ga$. We show the contrapositive: we suppose that SP has a non-losing (resp. winning) strategy $\strat$ on $\Ga'$, and we show that SP also has a non-losing (resp. winning) strategy on $\Ga$. Note that, since SP has a non-losing strategy on $\Ga'$, item {\em (i)} does not hold, so item {\em (ii)} does i.e. $y \in e \implies x \in e$ for every edge $e$ in $\Ga$.
    
    The idea is the same as in the proof of Lemma \ref{lem:greedy}: playing on $\Ga$, SP simply follows the strategy $\strat$, except if FP has just picked $x$ or $y$ in which case SP picks the other (or, if $x$ and $y$ are the only unpicked vertices, then SP picks one arbitrarily). 
    Since $\strat$ is a non-losing strategy for SP on $\Ga'$, FP cannot fill a red edge of $\Ga$, with the same proof as for Lemma \ref{lem:greedy}. Now, suppose that $\strat$ is a winning strategy for SP on $\Ga'$, which implies that SP eventually fills some blue edge $e'$ of $\Ga'$. We must check that SP also fills some blue edge of $\Ga$. This is straightforward: the blue edge $e'$ of $\Ga'$ necessarily stems from some blue edge $e$ of $\Ga$, moreover $x \not\in e$ (otherwise FP would have killed $e$ by picking $x$, so $e'$ would not have existed in $\Ga'$), therefore $y \not\in e$ which means that $e'=e$.
\end{proof}

\subsection{The Geography game}\strut
\indent An instance of \textsc{GeneralizedGeography} is a triple $(N,A,s)$ where $(N,A)$ is a simple digraph and $s \in N$. Two players, Alice and Bob, play the following game. Alice starts the game by placing the token on the prescribed node $s$. After that, a move for a player consists in moving the token from its current node to an out-neighbor of this node: Bob moves the token away from $s$, then Alice makes a move, then Bob makes a move, and so on. Whoever first moves the token to a node that had already been visited (i.e. that had already hosted the token previously) loses the game. The decision problem \textsc{GeneralizedGeography}, which asks whether Alice has a winning strategy, is {\PSPACE}-complete \cite{Geography}. More precisely:

\begin{theorem}\label{the:geography}
    \textup{\textsc{GeneralizedGeography}} is {\PSPACE}-complete, even when restricted to instances $(N,A,s)$ with the following properties:
    \begin{enumerate}[label={\textup{(\alph*)}},noitemsep,nolistsep]
		\item The simple digraph $(N,A)$ is bipartite;
        \item There is no cycle (oriented or not) of length less than 6 in $(N,A)$;
        \item The node $s$ has in-degree 0 and out-degree 1;
        \item Each node $v \neq s$ has either: in-degree 1 and out-degree 1, or in-degree 2 and out-degree 1, or in-degree 1 and out-degree 2.
	\end{enumerate}
\end{theorem}

\begin{proof}
	Properties (b) and (c) are additions compared to the construction provided in \cite[Corollary p. 396]{Geography}. First of all, since $s$ may have out-degree 2 in \cite{Geography}, we obtain property (c) by adding two nodes $s_0$ and $u$ along with the two arcs $\overrightarrow{s_0u}$ and $\overrightarrow{us}$, where $s_0$ is the new designated starting node. This fix is neutral for the game since, at the beginning of the game, Bob has no choice but to move the token from $s_0$ to $u$, and then Alice has no choice but to move the token from $u$ to $s$, and from there it is as if a fresh game started on the original instance. Second, we obtain property (b) by subdividing each arc twice, which triples the length of a shortest cycle (oriented or not). This fix is also neutral for the game, by an analogous argument.
\end{proof}

In the rest of this paper, the game restricted to instances that satisfy properties (a)--(b)--(c)--(d) from Theorem \ref{the:geography} will be referred to as the {\em Geography game}. We note that (b) will only be useful in Section \ref{section5}. 

\section{{\PSPACE}-completeness of 4-uniform Maker-Breaker games}\label{section3}\strut
\indent In this section, we are going to show the following result.

\begin{theorem}\label{the:mainMB}
    Deciding whether Maker has a winning strategy for the Maker-Breaker game on a 4-uniform hypergraph is a {\sf PSPACE}-complete problem.
\end{theorem}

Since membership in {\sf PSPACE} is well documented for Maker-Breaker games \cite{Schaefer}, it suffices to show {\sf PSPACE}-hardness.

\subsection{Idea of the construction}\strut
\indent We use the same idea as in \cite{Finn}. Given an instance $(N,A,s)$ of the Geography game, we will construct an associated instance of the Maker-Breaker game, which is a 4-uniform hypergraph, as follows. For each node $v \in N$, we define a {\em gadget} hypergraph $H_v$. For each arc $a=\overrightarrow{vw}$, there is a {\em junction pair} of vertices $\{p_a,q_a\}$ which only appear in $H_v$ (as an {\em output pair}) and in $H_w$ (as an {\em input pair}): we say that $p_a$ and $q_a$ are each other's {\em twin}. All vertices that are not junction vertices only appear in a single gadget, and we call them {\em interior vertices}. As such, the connections between the different gadgets are purely provided by the junction vertices. The hypergraph $H$ is simply the union of all these gadgets, i.e. $V(H)=\bigcup_{v \in N} V(H_v)$ and $E(H)=\bigcup_{v \in N} E(H_v)$. See Figure \ref{fig:example} for an illustration.

\begin{figure}[h]
\begin{minipage}{.08\linewidth}
\quad\quad
\end{minipage}
\begin{minipage}{.25\linewidth}
\begin{tikzpicture}
	\tikzset{noeud/.style={draw, circle, minimum height=0.45cm}}
    \node[noeud, label=center:$s$] (S) at (0,0) {};
    \node[noeud, label=center:$v_1$] (V1) at (1.5,0) {};
    \node[noeud, label=center:$v_2$] (V2) at (3.0,0) {};
    \node[noeud, label=center:$v_3$] (V3) at (2.25,-1.5) {};
    \node[noeud, label=center:$v_4$] (V4) at (3.75,-1.5) {};
    \node[noeud, label=center:$v_5$] (V5) at (1.5,-3.0) {};
    \node[noeud, label=center:$v_6$] (V6) at (3.0,-3.0) {};
    \node[noeud, label=center:$v_7$] (V7) at (2.25,-4.5) {};
    \node[noeud, label=center:$v_8$] (V8) at (4.5,-4.5) {};
    \draw[->] (S) -- node[above] {$a$} ++ (V1);
    \draw[->] (V1) -- node[above] {$b$} ++ (V2);
    \draw[->] (V2) -- node[left] {$c$} ++ (V3);
    \draw[->] (V2) -- node[right] {$d$} ++ (V4);
    \draw[->] (V3) -- node[left] {$e$} ++ (V5);
    \draw[->] (V3) -- node[left] {$f$} ++ (V6);
    \draw[->] (V4) -- node[right] {$g$} ++ (V6);
    \draw[->] (V5) -- node[left] {$h$} ++ (V7);
    \draw[->] (V6) -- node[right] {$i$} ++ (V7);
    \draw[->] (V7) -- node[below] {$j$} ++ (V8);
    \draw[->] (V8) -- node[right] {$k$} ++ (V4);
    \end{tikzpicture}
\end{minipage}
\begin{minipage}{.12\linewidth}
\quad\quad
\end{minipage}
\begin{minipage}{.55\linewidth}
\includegraphics[scale=.55]{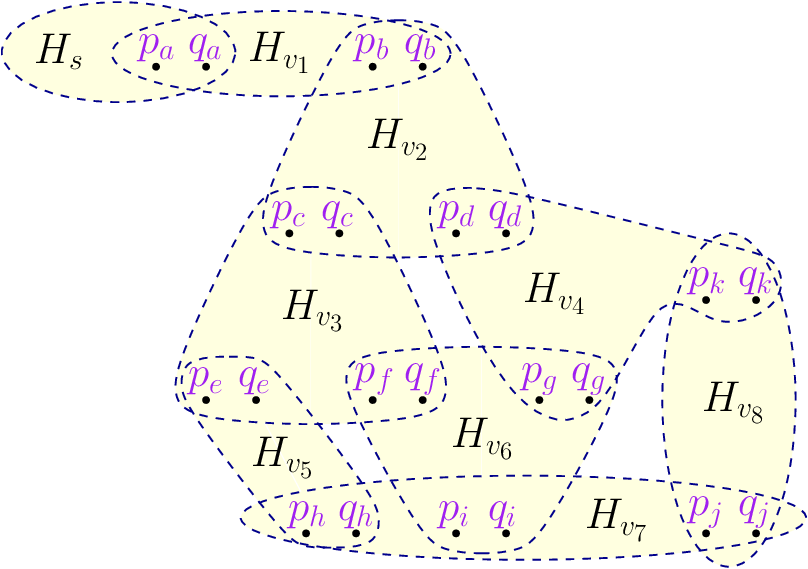}
\end{minipage}
\caption{Left: an example Geography instance. Right: a schematic representation of the corresponding Maker-Breaker game.}\label{fig:example}
\end{figure}
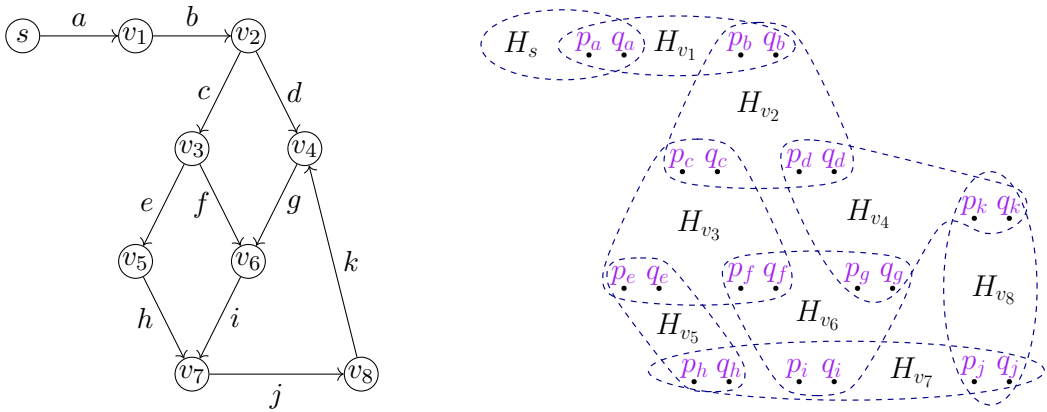

We define {\em regular play} for the Maker-Breaker game on $H$, referring to a very restricted set of strategies for both players which we will prove to be optimal. Regular play goes as follows. At every point, there is a single {\em active node}, which is defined as the node $s$ at the start. Each single move in the Geography game, where the player moves the token away from some node $v$, will correspond to a {\em sequence} of moves inside the gadget $H_v$ in the Maker-Breaker game, starting with a Maker move and ending with a Breaker move. That sequence goes through the gadget $H_v$, entering it via an input pair (except for $v=s$ as the gadget $H_s$ has no input pair) and exiting it via an output pair of which Maker picks both vertices. This output pair is the input pair of a unique gadget, whose corresponding node becomes the new active node. That sequence is virtually unique, except if $v$ has out-degree 2, in which case the player gets to choose between two sequences exiting via one output pair or the other, corresponding to the two options in the Geography game. The bipartiteness of the Geography digraph is key here. Indeed, it allows us to associate each node (and each gadget as a consequence) to a specific player in advance, Maker or Breaker, and to design the gadgets so that each choice between two sequences is made by the correct player: Maker if Alice makes the move, or Breaker if Bob makes the move. Regular play ends after the first time that a node that had already been active before gets activated again. Note that this can only happen for nodes of in-degree 2.

At the end of a sequence of regular play, the gadget that the players have gone through is ``cleared'' in the sense that Breaker has killed all its edges, apart from one specific case. For each node $v$ of in-degree 2 that can only be entered by Bob, we design the gadget $H_v$ so that an edge of the form $\{p,q,z\}$ survives in the updated gadget at the end of the sequence, where $\{p,q\}$ is the unused input pair of $H_v$. This way, if the gadget is entered again (necessarily via that input pair $\{p,q\}$), which corresponds to Bob losing the Geography game by reentering the node $v$, Maker will pick $p$ and $q$ as per regular play and thus win the game by picking $z$ as her next move. This is illustrated in Figure \ref{fig:example2}, corresponding to the situation after the moves $s \to v_1 \to v_2 \to v_4 \to v_6 \to v_7 \to v_8$ in the Geography instance from Figure \ref{fig:example}. In the Geography game, the next move will be Bob reentering $v_4$ and losing. In the Maker-Breaker game, the next sequence of regular play will take place inside $H_{v_8}$, with Maker picking $p_k$ and $q_k$ at the end, thus setting up her win on the next move where she will pick $z_k$. On the contrary, in the case where Alice loses, a fully cleared gadget is reentered, which causes Maker to relinquish the initiative and Breaker to win the game using a pairing strategy.

\begin{figure}[h]
\begin{minipage}{.08\linewidth}
\quad\quad
\end{minipage}
\begin{minipage}{.25\linewidth}
\begin{tikzpicture}
	\tikzset{noeud/.style={draw, circle, minimum height=0.45cm}}
	\tikzset{noeud_visite/.style={draw, fill=mygray, circle, minimum height=0.45cm}}
    \node[noeud_visite, label=center:$s$] (S) at (0,0) {};
    \node[noeud_visite, label=center:$v_1$] (V1) at (1.5,0) {};
    \node[noeud_visite, label=center:$v_2$] (V2) at (3.0,0) {};
    \node[noeud, label=center:$v_3$] (V3) at (2.25,-1.5) {};
    \node[noeud_visite, label=center:$v_4$] (V4) at (3.75,-1.5) {};
    \node[noeud, label=center:$v_5$] (V5) at (1.5,-3.0) {};
    \node[noeud_visite, label=center:$v_6$] (V6) at (3.0,-3.0) {};
    \node[noeud_visite, label=center:$v_7$] (V7) at (2.25,-4.5) {};
    \node[noeud_visite, label=center:$v_8$] (V8) at (4.5,-4.5) {};
    \draw[->,line width=1pt] (S) -- node[above] {$a$} ++ (V1);
    \draw[->,line width=1pt] (V1) -- node[above] {$b$} ++ (V2);
    \draw[->] (V2) -- node[left] {$c$} ++ (V3);
    \draw[->,line width=1pt] (V2) -- node[right] {$d$} ++ (V4);
    \draw[->] (V3) -- node[left] {$e$} ++ (V5);
    \draw[->] (V3) -- node[left] {$f$} ++ (V6);
    \draw[->,line width=1pt] (V4) -- node[right] {$g$} ++ (V6);
    \draw[->] (V5) -- node[left] {$h$} ++ (V7);
    \draw[->,line width=1pt] (V6) -- node[right] {$i$} ++ (V7);
    \draw[->,line width=1pt] (V7) -- node[below] {$j$} ++ (V8);
    \draw[->] (V8) -- node[right] {$k$} ++ (V4);
    \end{tikzpicture}
\end{minipage}
\begin{minipage}{.12\linewidth}
\quad\quad
\end{minipage}
\begin{minipage}{.55\linewidth}
\includegraphics[scale=.55]{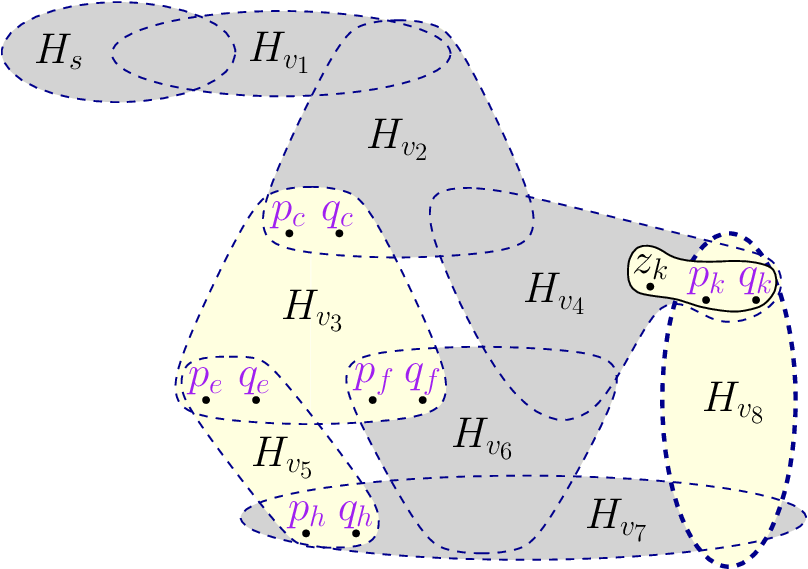}
\end{minipage}
\caption{Illustration of a scenario where Maker wins, pictured before the final sequence of regular play. Cleared areas are grayed out.}\label{fig:example2}
\end{figure}
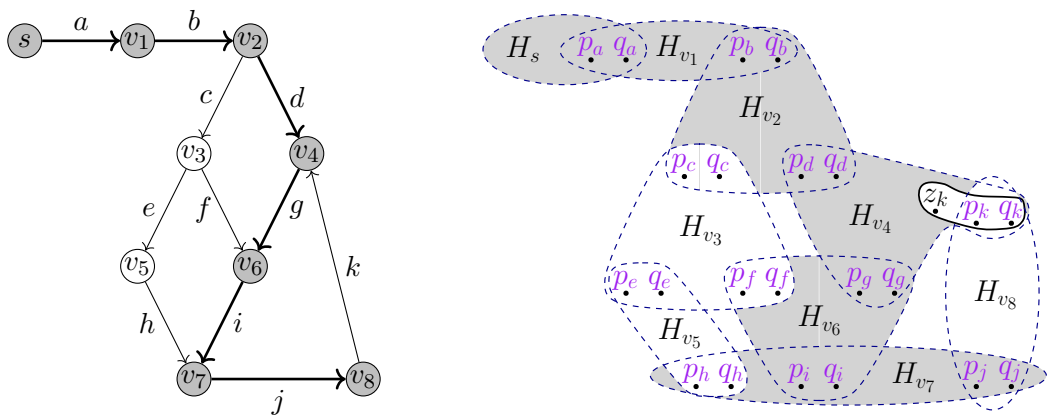

Pairings will also be key to our proof that regular play is optimal. A tricky aspect of hardness reductions for Maker-Breaker games is that, though it is easy to force Breaker's moves using short-term Maker threats, it is more difficult to constrain Maker's moves. Our construction is designed so that, in all nontrivial cases where Maker deviates from regular play, Breaker can punish her by playing a move after which the updated hypergraph admits a pairing. To exhibit this pairing, the idea is to find a pairing in each (updated) gadget, and take the union of all these. Mind that this does not work if some junction vertex is paired with different vertices in the two gadgets it appears in. For this reason, we introduce the following definitions. A pair of vertices is called {\em clean} if it is a junction pair or if both its vertices are interior vertices of the same gadget, otherwise it is called {\em mixed}. A pairing is deemed {\em clean} if all its pairs are clean. Clean pairings of (updated) gadgets are compatible with each other, with the junction pairs stapling them together to form a global pairing. Whenever Maker violates regular play by picking a junction vertex $p$, we use its twin $q$ in the pairing of the gadget in which $\{p,q\}$ is an input pair, while we manage without $p$ and $q$ in the pairing of the gadget in which $\{p,q\}$ is an output pair. This is illustrated in Figure \ref{fig:example3}.

\begin{figure}[h]
\centering
\includegraphics[scale=.55]{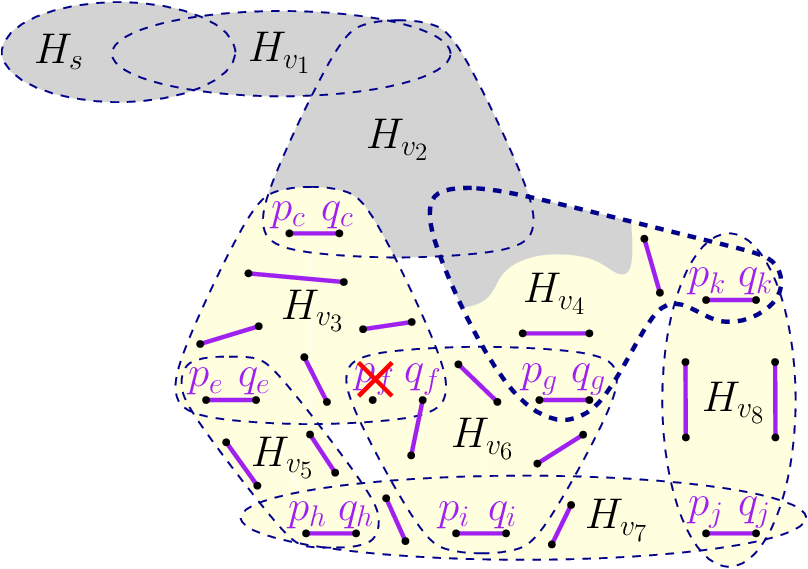}
\caption{Schematic representation of the pairing obtained after a round in which Maker has picked an irregular vertex ($p_f$ here, while $v_4$ was the active node). All pairings of updated gadgets are clean, apart from that of $H_{v_6}$ since it pairs $q_f$ with some interior vertex.}\label{fig:example3}
\end{figure}

\subsection{The gadgets and their pairings}\strut 
\indent Consider a 2-coloring of the digraph $(N,A)$. Similarly to what is done in \cite{Finn}, we partition the set $N$ into 7 types, depending on the node's color as well as its in-degree and out-degree:
\begin{itemize}[nolistsep,noitemsep]
    \item $B_{0,1}=\{s\}$ contains the only node of in-degree 0 and out-degree 1;
    \item $B_{1,1}$ is the set of all nodes with same color as $s$, in-degree 1 and out-degree 1;
    \item $B_{2,1}$ is the set of all nodes with same color as $s$, in-degree 2 and out-degree 1;
    \item $B_{1,2}$ is the set of all nodes with same color as $s$, in-degree 1 and out-degree 2;
    \item $M_{1,1}$ is the set of all nodes with different color than $s$, in-degree 1 and out-degree 1;
    \item $M_{2,1}$ is the set of all nodes with different color than $s$, in-degree 2 and out-degree 1;
    \item $M_{1,2}$ is the set of all nodes with different color than $s$, in-degree 1 and out-degree 2.
\end{itemize}
For instance, in the exemple from Figure \ref{fig:example}, we have: $B_{0,1}=\{s\}$, $B_{1,1}=\{v_5,v_8\}$, $B_{2,1}=\{v_6\}$, $B_{1,2}=\{v_2\}$, $M_{1,1}=\{v_1\}$, $M_{2,1}=\{v_4,v_7\}$, $M_{1,2}=\{v_3\}$.

The gadgets and their associated sequence(s) of regular play are given in Tables \ref{tab:tableau1}--\ref{tab:tableau2}--\ref{tab:tableau3}. 
For a generic node $v$ of each possible type, these tables provide:
\begin{itemize}
	\item[--] A graphical representation of the gadget $H_v$ (which is only schematic for $v=s$), where edges are colored for readability only. We list all the edges anyway, in case the drawing is not clear enough. Throughout the proof, every time we consider a specific gadget, we will freely refer to its vertices by these names: we point out that, even though every gadget has a vertex named ``$y_1$'' for instance, there will be no ambiguity in practice as to which gadget we are referring to at any given moment.
	\item[--] An alternative representation of $H_v$ through a graph whose 4-cycles are the edges of $H_v$. This will only be used in Section \ref{section5} and can be ignored for the moment.
	\item[--] The sequence(s) of regular play that takes place inside $H_v$ when $v$ is the active node. Maker's moves are written in red and Breaker's moves are written in blue (for the reader using a grayscale version of this paper, simply recall that Maker makes the first move in every sequence). An arrow labelled ``g'' (resp. ``2-g'', resp. ``3-g'') signifies a 1-greedy move (resp. 2-greedy move, resp. 3-greedy move) by Maker as per Definition \ref{def:greedy}, and an arrow labelled ``f'' signifies a {\em forced move} by Breaker, meaning that Maker was threatening a win in one move.
	
	For example, let us detail the case $v \in M_{2,1}$ in Table \ref{tab:tableau3}. Suppose that $u_1$ was active immediately prior i.e. $v$ has been entered through the arc $a$, meaning that Maker has already picked $p_a$ and $q_a$ while $u_1$ was active. Regular play on $H_v$ starts as follows: Maker picks $x_1$, Breaker picks $z_a$, Maker picks $p_c$, Breaker picks $y_1$, Maker picks $x_2$. Now, Breaker has a (strategically inconsequential) choice between three possible moves, all of which comply with regular play. If Breaker picks $z_1$ (resp. $z_2$, resp. $z_3$), then Maker picks $z_2$ (resp. $z_3$, resp. $z_1$) and Breaker picks $z_3$ (resp. $z_1$, resp. $z_2$). After that, Maker picks $q_c$ and Breaker picks $y_2$ to end this sequence of regular play. The case $v \in B_{2,1}$ is analogous. The only other times when regular play does not suggest a unique move happen for $v \in B_{1,2}$ and $v \in M_{1,2}$ of course, as Breaker or Maker respectively must make a choice between the two out-neighbors at some point.
	
\end{itemize}

\begin{table}[H]

    \centering
    
    \begin{tabular}{ | >{\centering\arraybackslash} m{1.3cm} | >{\centering\arraybackslash} m{5.6cm} | >{\centering\arraybackslash} m{5.4cm} | >{\centering\arraybackslash} m{2.45cm} |}
    
    \hline
    
    Type of the node $v$
    &
    Gadget $H_v$
    &
    Graph version of $H_v$ for the vertex-$C_4$-game
    & Sequence of regular play when $v$ is active
    
    \\\hline
    
    {$v$}{$\,=\,$}{$s$}
    
    {$\in\,$}{$B_{0,1}$}
    
	\begin{tikzpicture}
	\tikzset{noeud/.style={draw, circle, minimum height=0.45cm}}
    \node[noeud, label=center:$v$] (V) at (0,0) {};
    \node[noeud, label=center:$w$] (W) at (0,-1) {};
    \draw[->] (V) -- node[right] {$a$} ++ (W);
    \draw (0, 0.5) node {\vphantom{a}};
    \end{tikzpicture}

    & \includegraphics[scale=.5]{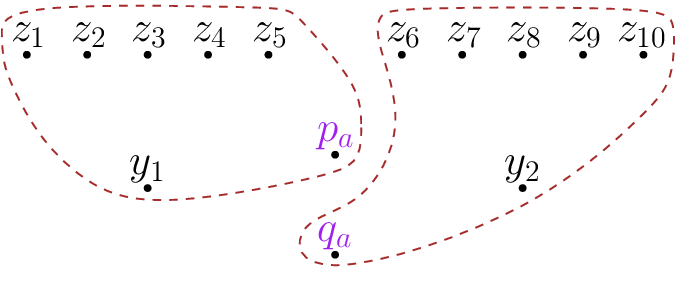}
    
    Edges: $\{p_a,y_1,z_i,z_j\}$ for all $i \neq j \in \{1,2,3,4,5\}$, and $\{q_a,y_2,z_i,z_j\}$ for all $i \neq j \in \{6,7,8,9,10\}$

    & \includegraphics[scale=.5]{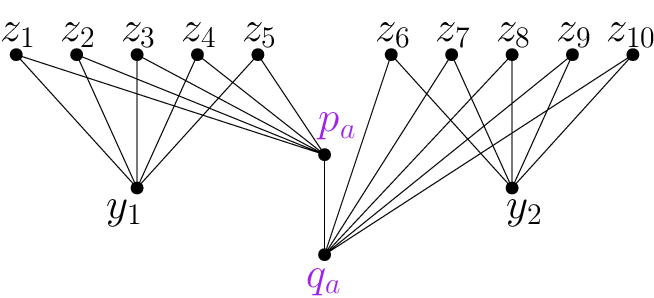}
    \,
    
    \,
    
    \,
    
    \,
    &
    \begin{tikzcd}[row sep = tiny , sep = tiny]
             \arrow{r}{\text{3-g}} & \textcolor{red}{p_a} \arrow{r}{} & \textcolor{blue}{y_1}
    \end{tikzcd}
    
    \begin{tikzcd}[row sep = tiny , sep = tiny]
             \arrow{r}{\text{3-g}} & \textcolor{red}{q_a} \arrow{r}{} & \textcolor{blue}{y_2}
    \end{tikzcd}
    \\\hline
    
    {$v\,$}{$\in\,$}{$B_{1,1}$}
    
    {$\cup\,$}{$M_{1,1}$}
    
	\begin{tikzpicture}
	\tikzset{noeud/.style={draw, circle, minimum height=0.45cm}}
	\node[noeud, label=center:$u$] (U) at (0,1) {};
    \node[noeud, label=center:$v$] (V) at (0,0) {};
    \node[noeud, label=center:$w$] (W) at (0,-1) {};
    \draw[->] (U) -- node[right] {$a$} ++ (V);
    \draw[->] (V) -- node[right] {$b$} ++ (W);
    \draw (0, 1.5) node {\vphantom{a}};
    \end{tikzpicture}

    & \includegraphics[scale=.5]{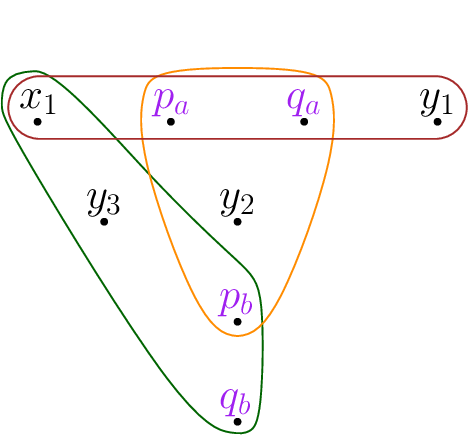}
    
    Edges: $\{p_a,q_a,x_1,y_1\}$, $\{p_a,q_a,p_b,y_2\}$, $\{p_b,q_b,x_1,y_3\}$
    
    & \includegraphics[scale=.5]{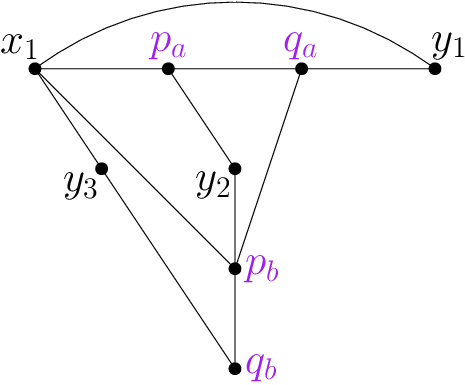}
    \,
    
    \,
    & 
    \begin{tikzcd}[row sep = tiny , sep = tiny]
             \arrow{r}{\text{g}} & \textcolor{red}{x_1} \arrow{r}{\text{f}} & \textcolor{blue}{y_1}
    \end{tikzcd}
    
    \begin{tikzcd}[row sep = tiny , sep = tiny]
             \arrow{r}{\text{g}} & \textcolor{red}{p_b} \arrow{r}{\text{f}} & \textcolor{blue}{y_2}
    \end{tikzcd}
    
    \begin{tikzcd}[row sep = tiny , sep = tiny]
             \arrow{r}{\text{g}} & \textcolor{red}{q_b} \arrow{r}{\text{f}} & \textcolor{blue}{y_3}
    \end{tikzcd}
    \\\hline
    
    \end{tabular}
    
    \caption{The gadgets and their associated sequence(s) of regular play (part 1/3).}\label{tab:tableau1}
    
\end{table}

\begin{table}[H]

    \centering
    
    \begin{tabular}{ | >{\centering\arraybackslash} m{1.6cm} | >{\centering\arraybackslash} m{5.1cm} | >{\centering\arraybackslash} m{3.1cm} | >{\centering\arraybackslash} m{4.85cm} |}
    
    \hline
    
    Type of the node $v$
    &
    Gadget $H_v$
    &
    Graph version of $H_v$ for the vertex-$C_4$-game
    & Sequence of regular play when $v$ is active
    
    \\\hline
    
    {$v\,$}{$\in\,$}{$B_{1,2}$}
    
	\begin{tikzpicture}
	\tikzset{noeud/.style={draw, circle, minimum height=0.45cm}}
	\node[noeud, label=center:$u$] (U) at (0,1) {};
    \node[noeud, label=center:$v$] (V) at (0,0) {};
    \node[noeud, label=center:$w_1$] (W1) at (-0.5,-1) {};
    \node[noeud, label=center:$w_2$] (W2) at (0.5,-1) {};
    \draw[->] (U) -- node[right] {$a$} ++ (V);
    \draw[->] (V) -- node[left] {$b$} ++ (W1);
    \draw[->] (V) -- node[right] {$c$} ++ (W2);
    \draw (0, 1.5) node {\vphantom{a}};
    \end{tikzpicture}
    & \includegraphics[scale=.5]{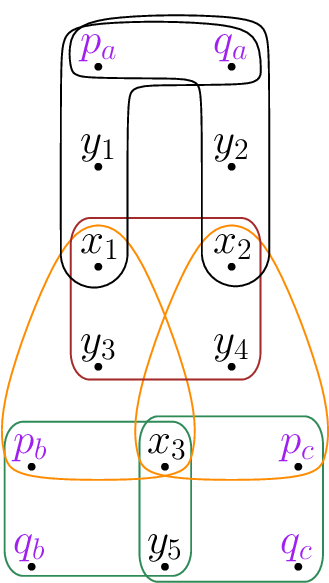}
    
    Edges: $\{p_a,q_a,x_1,y_1\}$, $\{p_a,q_a,x_2,y_2\}$, $\{x_1,x_2,y_3,y_4\}$, $\{p_b,x_1,x_3,y_3\}$, $\{p_c,x_2,x_3,y_4\}$, $\{p_b,q_b,x_3,y_5\}$, $\{p_c,q_c,x_3,y_5\}$
    & \includegraphics[scale=.5]{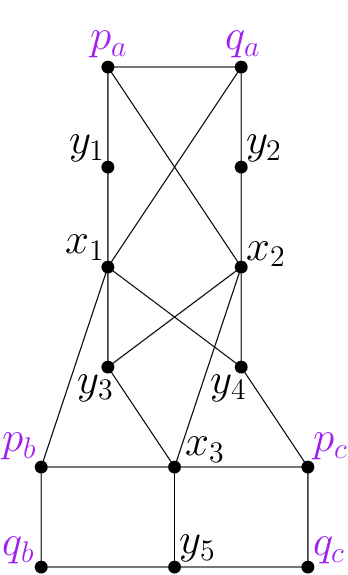}
    \,
    
    \,
    
    \,
    
    \,
    &
    \begin{tikzcd}[row sep = tiny , sep = tiny]
             \arrow{r}{\text{g}} & \textcolor{red}{x_1} \arrow{r}{\text{f}} & \textcolor{blue}{y_1} \arrow{r}{\text{g}} & \textcolor{red}{x_2} \arrow{r}{\text{f}} & \textcolor{blue}{y_2} \arrow{r}{} & \textcolor{red}{x_3}
    \end{tikzcd}
    
    \medskip
    followed by, if Breaker chooses to activate $w_1$:
    
    \begin{tikzcd}[row sep = tiny , sep = tiny]
             \arrow{r}{} & \textcolor{blue}{y_4} \arrow{r}{\text{g}} & \textcolor{red}{p_b} \arrow{r}{\text{f}} & \textcolor{blue}{y_3} \arrow{r}{} & \textcolor{red}{q_b} \arrow{r}{\text{f}} & \textcolor{blue}{y_5}
	\end{tikzcd}
	
	\medskip
	or followed by, if Breaker chooses to activate $w_2$:
	
    \begin{tikzcd}[row sep = tiny , sep = tiny]
             \arrow{r}{} & \textcolor{blue}{y_3} \arrow{r}{\text{g}} & \textcolor{red}{p_c} \arrow{r}{\text{f}} & \textcolor{blue}{y_4} \arrow{r}{} & \textcolor{red}{q_c} \arrow{r}{\text{f}} & \textcolor{blue}{y_5}
	\end{tikzcd}
    \\\hline
    
    {$v\,$}{$\in\,$}{$M_{1,2}$}
    
	\begin{tikzpicture}
	\tikzset{noeud/.style={draw, circle, minimum height=0.45cm}}
	\node[noeud, label=center:$u$] (U) at (0,1) {};
    \node[noeud, label=center:$v$] (V) at (0,0) {};
    \node[noeud, label=center:$w_1$] (W1) at (-0.5,-1) {};
    \node[noeud, label=center:$w_2$] (W2) at (0.5,-1) {};
    \draw[->] (U) -- node[right] {$a$} ++ (V);
    \draw[->] (V) -- node[left] {$b$} ++ (W1);
    \draw[->] (V) -- node[right] {$c$} ++ (W2);
    \draw (0, 1.5) node {\vphantom{a}};
    \end{tikzpicture}
    & \includegraphics[scale=.5]{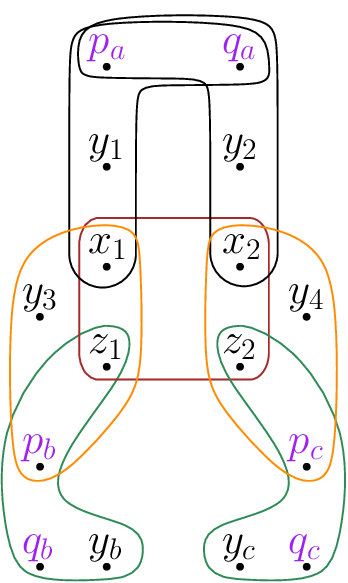}
    
    Edges: $\{p_a,q_a,x_1,y_1\}$, $\{p_a,q_a,x_2,y_2\}$, $\{x_1,x_2,z_1,z_2\}$, $\{p_b,x_1,y_3,z_1\}$, $\{p_c,x_2,y_4,z_2\}$, $\{p_b,q_b,y_b,z_1\}$, $\{p_c,q_c,y_c,z_2\}$
    & \includegraphics[scale=.5]{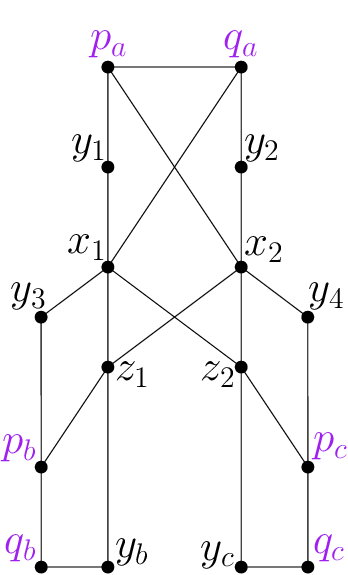}
    \,
    
    \,
    
    \,
    
    \,
    & 
    \begin{tikzcd}[row sep = tiny , sep = tiny]
             \arrow{r}{\text{g}} & \textcolor{red}{x_1} \arrow{r}{\text{f}} & \textcolor{blue}{y_1} \arrow{r}{\text{g}} & \textcolor{red}{x_2} \arrow{r}{\text{f}} & \textcolor{blue}{y_2}
    \end{tikzcd}
    
    \medskip      
    followed by, if Maker chooses to activate $w_1$:
    
    \begin{tikzcd}[row sep = tiny , sep = tiny]
             \arrow{r}{} & \textcolor{red}{z_1} \arrow{r}{\text{f}} & \textcolor{blue}{z_2}
	\end{tikzcd}
	
	\begin{tikzcd}[row sep = tiny , sep = tiny]
             \arrow{r}{\text{g}} & \textcolor{red}{p_b} \arrow{r}{\text{f}} & \textcolor{blue}{y_3} \arrow{r}{\text{g}} & \textcolor{red}{q_b} \arrow{r}{\text{f}} & \textcolor{blue}{y_b}
	\end{tikzcd}
	
	\medskip
	or followed by, if Maker chooses to activate $w_2$:
	
    \begin{tikzcd}[row sep = tiny , sep = tiny]
             \arrow{r}{} & \textcolor{red}{z_2} \arrow{r}{\text{f}} & \textcolor{blue}{z_1}
	\end{tikzcd}
	
	\begin{tikzcd}[row sep = tiny , sep = tiny]
             \arrow{r}{\text{g}} & \textcolor{red}{p_c} \arrow{r}{\text{f}} & \textcolor{blue}{y_4} \arrow{r}{\text{g}} & \textcolor{red}{q_c} \arrow{r}{\text{f}} & \textcolor{blue}{y_c}
	\end{tikzcd}
    \\\hline
    
    \end{tabular}
    
    \caption{The gadgets and their associated sequence(s) of regular play (part 2/3).}\label{tab:tableau2}
    
\end{table}

\begin{table}[H]

    \centering
    
    \begin{tabular}{ | >{\centering\arraybackslash} m{1.4cm} | >{\centering\arraybackslash} m{4.95cm} | >{\centering\arraybackslash} m{4.7cm} | >{\centering\arraybackslash} m{3.6cm} |}
    
    \hline
    
    Type of the node $v$
    &
    Gadget $H_v$
    &
    Graph version of $H_v$ for the vertex-$C_4$-game
    & Sequence of regular play when $v$ is active
    
    \\\hline
    
    {$v\,$}{$\in\,$}{$B_{2,1}$}
    
	\begin{tikzpicture}
	\tikzset{noeud/.style={draw, circle, minimum height=0.45cm}}
	\node[noeud, label=center:$u_1$] (U1) at (-0.5,1) {};
	\node[noeud, label=center:$u_2$] (U2) at (0.5,1) {};
    \node[noeud, label=center:$v$] (V) at (0,0) {};
    \node[noeud, label=center:$w$] (W) at (0,-1) {};
    \draw[->] (U1) -- node[left] {$a$} ++ (V);
    \draw[->] (U2) -- node[right] {$b$} ++ (V);
    \draw[->] (V) -- node[right] {$c$} ++ (W);
    \draw (0, 1.5) node {\vphantom{a}};
    \end{tikzpicture}
    & \includegraphics[scale=.5]{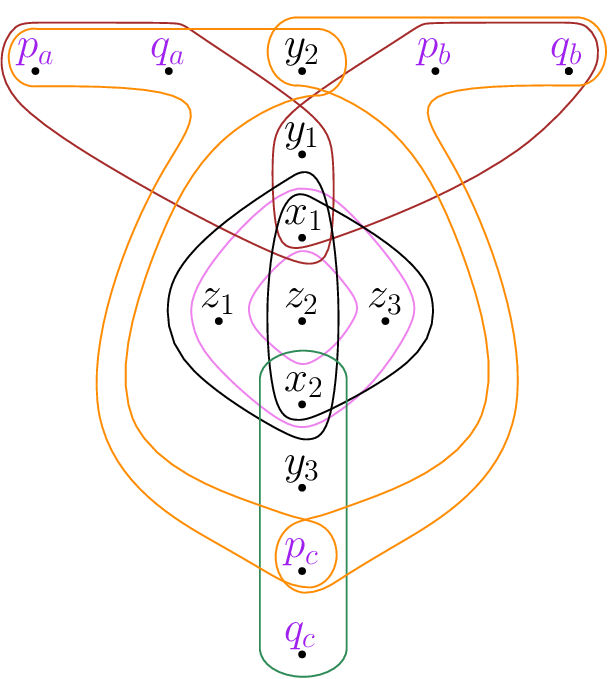}
    
    Edges: $\{p_a,q_a,x_1,y_1\}$, $\{p_b,q_b,x_1,y_1\}$, $\{p_a,q_a,p_c,y_2\}$,  $\{p_b,q_b,p_c,y_2\}$, $\{x_1,x_2,z_1,z_2\}$, $\{x_1,x_2,z_1,z_3\}$, $\{x_1,x_2,z_2,z_3\}$, $\{p_c,q_c,x_2,y_3\}$
    & \includegraphics[scale=.5]{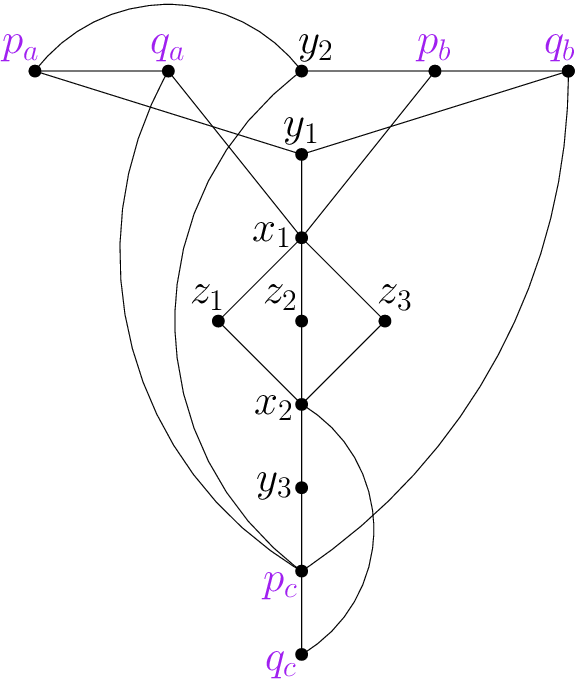}
    \,
    
    \,
    
    \,
    
    \,
    
    \,
    &
    \begin{tikzcd}[row sep = tiny , sep = tiny]
             \arrow{r}{\text{g}} & \textcolor{red}{x_1} \arrow{r}{\text{f}} & \textcolor{blue}{y_1}
	\end{tikzcd}
	
    \begin{tikzcd}[row sep = tiny , sep = tiny]
             \arrow{r}{\text{g}} & \textcolor{red}{p_c} \arrow{r}{\text{f}} & \textcolor{blue}{y_2} \arrow{r}{\text{2-g}} & \textcolor{red}{x_2}
	\end{tikzcd}
	
    \begin{tikzcd}[row sep = tiny , sep = tiny]
             & \textcolor{blue}{z_1} \arrow{r}{\text{g}} & \textcolor{red}{z_2} \arrow{r}{\text{f}} & \textcolor{blue}{z_3} \arrow[bend left]{dr}{\text{g}} & \\
             \arrow[bend left]{ur}[swap]{} \arrow[bend right]{dr}[swap]{} \arrow{r}{} & \textcolor{blue}{z_2} \arrow{r}{\text{g}} & \textcolor{red}{z_3} \arrow{r}{\text{f}} & \textcolor{blue}{z_1} \arrow{r}{\text{g}} & \textcolor{red}{q_c} \\
            & \textcolor{blue}{z_3} \arrow{r}{\text{g}} & \textcolor{red}{z_1} \arrow{r}{\text{f}} & \textcolor{blue}{z_2} \arrow[bend right]{ur}[swap]{\text{g}} &
\end{tikzcd}

\begin{tikzcd}[row sep = tiny , sep = tiny]
             \arrow{r}{\text{f}} & \textcolor{blue}{y_3}
	\end{tikzcd}
	
	\medskip
	($w$ is the new active node)
    \\\hline
    
    {$v\,$}{$\in\,$}{$M_{2,1}$}
    
	\begin{tikzpicture}
	\tikzset{noeud/.style={draw, circle, minimum height=0.45cm}}
	\node[noeud, label=center:$u_1$] (U1) at (-0.5,1) {};
	\node[noeud, label=center:$u_2$] (U2) at (0.5,1) {};
    \node[noeud, label=center:$v$] (V) at (0,0) {};
    \node[noeud, label=center:$w$] (W) at (0,-1) {};
    \draw[->] (U1) -- node[left] {$a$} ++ (V);
    \draw[->] (U2) -- node[right] {$b$} ++ (V);
    \draw[->] (V) -- node[right] {$c$} ++ (W);
    \draw (0, 1.5) node {\vphantom{a}};
    \end{tikzpicture}
    & \includegraphics[scale=.5]{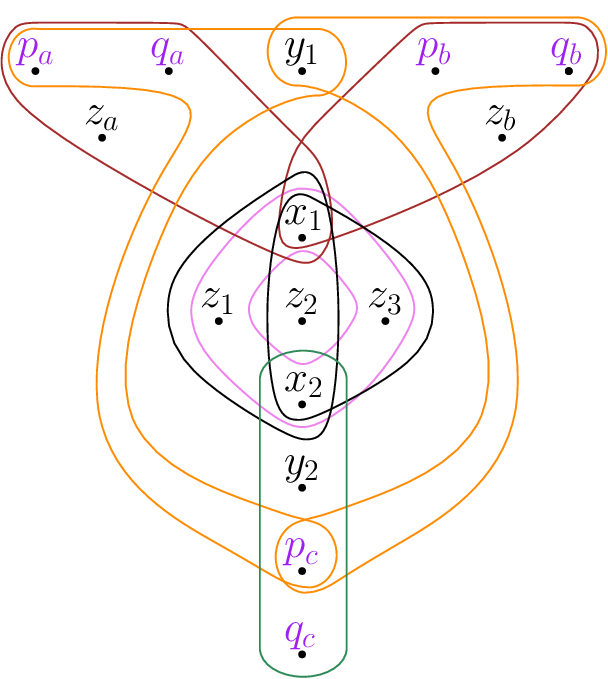}
    
    Edges: $\{p_a,q_a,x_1,z_a\}$, $\{p_b,q_b,x_1,z_b\}$, $\{p_a,q_a,p_c,y_1\}$,  $\{p_b,q_b,p_c,y_1\}$, $\{x_1,x_2,z_1,z_2\}$, $\{x_1,x_2,z_1,z_3\}$, $\{x_1,x_2,z_2,z_3\}$, $\{p_c,q_c,x_2,y_2\}$
    & \includegraphics[scale=.5]{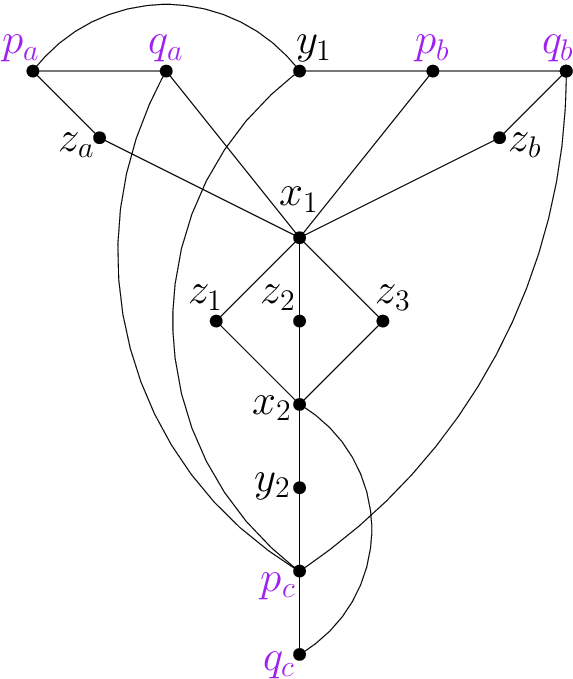}
    \,
    
    \,
    
    \,
    
    \,
    
    \,
    & if $u_1$ was active immediately prior: 
    
    \begin{tikzcd}[row sep = tiny , sep = tiny]
             \arrow{r}{\text{g}} & \textcolor{red}{x_1} \arrow{r}{\text{f}} & \textcolor{blue}{z_a}
	\end{tikzcd}
	
	\medskip
	or, if $u_2$ was active immediately prior:
	
    \begin{tikzcd}[row sep = tiny , sep = tiny]
             \arrow{r}{\text{g}} & \textcolor{red}{x_1} \arrow{r}{\text{f}} & \textcolor{blue}{z_b}
	\end{tikzcd}
	
	\medskip
    followed by, in both cases:
    
    \begin{tikzcd}[row sep = tiny , sep = tiny]
             \arrow{r}{\text{g}} & \textcolor{red}{p_c} \arrow{r}{\text{f}} & \textcolor{blue}{y_1} \arrow{r}{\text{2-g}} & \textcolor{red}{x_2}
	\end{tikzcd}
	
    \begin{tikzcd}[row sep = tiny , sep = tiny]
             & \textcolor{blue}{z_1} \arrow{r}{\text{g}} & \textcolor{red}{z_2} \arrow{r}{\text{f}} & \textcolor{blue}{z_3} \arrow[bend left]{dr}{\text{g}} & \\
             \arrow[bend left]{ur}[swap]{} \arrow[bend right]{dr}[swap]{} \arrow{r}{} & \textcolor{blue}{z_2} \arrow{r}{\text{g}} & \textcolor{red}{z_3} \arrow{r}{\text{f}} & \textcolor{blue}{z_1} \arrow{r}{\text{g}} & \textcolor{red}{q_c} \\
            & \textcolor{blue}{z_3} \arrow{r}{\text{g}} & \textcolor{red}{z_1} \arrow{r}{\text{f}} & \textcolor{blue}{z_2} \arrow[bend right]{ur}[swap]{\text{g}} &
\end{tikzcd}

\begin{tikzcd}[row sep = tiny , sep = tiny]
             \arrow{r}{\text{f}} & \textcolor{blue}{y_2}
	\end{tikzcd}
	
	\medskip
	($w$ is the new active node)
    \\\hline
    
    \end{tabular}
    
    \caption{The gadgets and their associated sequence(s) of regular play (part 3/3).}\label{tab:tableau3}
    
\end{table}

The next claim provides a (clean) pairing for each intact gadget, as well as a (not necessarily clean) pairing for each gadget in which a given vertex $z$ is not usable. The latter is important, as it shows that a pairing still exists if Maker has picked a single vertex (namely, $z$) inside a gadget, possibly as a violation of regular play.

\begin{claim}\label{cla:pairing}
    Let $v \in N \setminus B_{0,1} = N \setminus \{s\}$. The gadget $H_v$ admits all the following pairings:
    \begin{itemize}[noitemsep,nolistsep]
        \item a clean pairing $\Pi(v)$;
        \item for each vertex $z \in V(H_v)$ which is not an input vertex of $H_v$: a clean pairing $\Pi(v,z)$ which does not use $z$;
        \item for each vertex $z \in V(H_v)$ which is an input vertex of $H_v$: a pairing $\Pi(v,z)$ which does not use $z$ and whose only mixed pair contains the twin of $z$ paired with an interior vertex.
    \end{itemize}
\end{claim} 

\begin{proof}

We refer to Tables \ref{tab:tableau1}--\ref{tab:tableau2}--\ref{tab:tableau3} for the names of the vertices in $H_v$ and arcs incident to $v$. Some cases will be omitted thanks to symmetries: for example, in the case $v \in M_{1,2}$, the pairing $\Pi(v,x_2)$ is defined analogously to the pairing $\Pi(v,x_1)$. 
    \begin{enumerate}[label={(\arabic*)}]
        \item Case $v \in B_{1,1} \cup M_{1,1}$. We define:
        \\ $\Pi(v) = \{ \{p_a,q_a\} , \{p_b,q_b\} \}$;
        \\ $\Pi(v,x_1) = \Pi(v,y_1) = \Pi(v,y_2) = \Pi(v,y_3) = \Pi(v)$;
        \\ $\Pi(v,p_a) = \{ \{q_a,y_2\} , \{p_b,q_b\} , \{x_1,y_1\} \}$;
        \\ $\Pi(v,p_b) = \Pi(v,q_b) = \{ \{p_a,q_a\} , \{x_1,y_3\} \}$.
        \item Case $v \in B_{1,2}$. We define:
        \\ $\Pi(v) = \{ \{p_a,q_a\} , \{p_b,q_b\} , \{p_c,q_c\} , \{x_1,y_3\} , \{x_2,y_4\} \}$;
        \\ $\Pi(v,x_3) = \Pi(v,y_1) = \Pi(v,y_5) = \Pi(v)$;
        \\ $\Pi(v,x_1) = \{ \{p_a,q_a\} , \{p_b,q_b\} , \{p_c,q_c\} , \{x_3,y_3\} , \{x_2,y_4\} \}$;
        \\ $\Pi(v,y_3) = \{ \{p_a,q_a\} , \{p_b,q_b\} , \{p_c,q_c\} , \{x_1,x_3\} , \{x_2,y_4\} \}$;
        \\ $\Pi(v,p_a) = \{ \{q_a,y_1\} , \{p_b,q_b\} , \{p_c,q_c\} , \{x_1,y_3\} , \{x_2,y_2\} , \{x_3,y_4\} \}$;
        \\ $\Pi(v,p_b) = \Pi(v,q_b) = \{ \{p_a,q_a\} , \{p_c,q_c\} , \{x_1,y_3\} , \{x_2,y_4\} , \{x_3,y_5\} \}$
        \\(Note that, in that last pairing, the pair $\{p_c,q_c\}$ is actually redundant. We add it out of principle, so that every pairing includes all allowed junction pairs.)
        \item Case $v \in M_{1,2}$. We define: 
        \\ $\Pi(v) = \{ \{p_a,q_a\} , \{p_b,q_b\} , \{p_c,q_c\} , \{x_1,z_1\} , \{x_2,z_2\} \}$;
        \\ $\Pi(v,y_1) = \Pi(v,y_3) = \Pi(v)$;
        \\ $\Pi(v,x_1) = \{ \{p_a,q_a\} , \{p_b,q_b\} , \{p_c,q_c\} , \{y_3,z_1\} , \{x_2,z_2\} \}$;
        \\ $\Pi(v,z_1) = \{ \{p_a,q_a\} , \{p_b,q_b\} , \{p_c,q_c\} , \{x_1,y_3\} , \{x_2,z_2\} \}$;
        \\ $\Pi(v,p_a) = \{ \{q_a,y_1\} , \{p_b,q_b\} , \{p_c,q_c\} , \{x_1,z_1\} , \{x_2,y_2\} , \{y_4,z_2\} \}$;
        \\ $\Pi(v,p_b) = \Pi(v,q_b) = \{ \{p_a,q_a\} , \{p_c,q_c\} , \{x_1,y_3\} , \{x_2,z_2\} , \{z_1,y_b\} \}$.
        \item Case $v \in B_{2,1}$. We define:
        \\ $\Pi(v) = \{ \{p_a,q_a\} , \{p_b,q_b\} , \{p_c,q_c\} , \{x_1,x_2\} \}$;
        \\ $\Pi(v,x_3) = \Pi(v,y_1) = \Pi(v,y_2) = \Pi(v,y_3) = \Pi(v,y_4) = \Pi(v,y_5) = \Pi(v)$;
        \\ $\Pi(v,x_1) = \{ \{p_a,q_a\} , \{p_b,q_b\} , \{p_c,q_c\} , \{x_2,z_1\} , \{z_2,z_3\} \}$;
        \\ $\Pi(v,x_2) = \{ \{p_a,q_a\} , \{p_b,q_b\} , \{p_c,q_c\} , \{x_1,z_1\} , \{z_2,z_3\} \}$;
        \\ $\Pi(v,p_a) = \{ \{q_a,y_2\} , \{p_b,q_b\} , \{p_c,q_c\} , \{x_1,y_1\} , \{x_2,z_1\} , \{z_2,z_3\} \}$;
        \\ $\Pi(v,p_c) = \Pi(v,q_c) = \{ \{p_a,q_a\} , \{p_b,q_b\} , \{x_1,z_1\} , \{z_2,z_3\}, \{x_2,y_3\} \}$.
        \item Case $v \in M_{2,1}$. We define:
        \\ $\Pi(v) = \{ \{p_a,q_a\} , \{p_b,q_b\} , \{p_c,q_c\} , \{x_1,x_2\} \}$;
        \\ $\Pi(v,y_1) = \Pi(v,y_2) = \Pi(v,z_1) = \Pi(v,z_2) = \Pi(v,z_3) = \Pi(v,z_a) = \Pi(v,z_b) = \Pi(v)$;
        \\ $\Pi(v,x_1) = \{ \{p_a,q_a\} , \{p_b,q_b\} , \{p_c,q_c\} , \{x_2,z_1\} , \{z_2,z_3\} \}$;
        \\ $\Pi(v,x_2) = \{ \{p_a,q_a\} , \{p_b,q_b\} , \{p_c,q_c\} , \{x_1,z_1\} , \{z_2,z_3\} \}$;
        \\ $\Pi(v,p_a) = \{ \{q_a,y_1\} , \{p_b,q_b\} , \{p_c,q_c\} , \{x_1,z_a\} , \{x_2,z_1\} , \{z_2,z_3\} \}$;
        \\ $\Pi(v,p_c) = \Pi(v,q_c) = \{ \{p_a,q_a\} , \{p_b,q_b\} , \{x_1,z_1\} , \{z_2,z_3\}, \{x_2,y_2\} \}$.
    \end{enumerate}
\noindent It is easily checked that these pairings have all the desired properties.
\end{proof}

\subsection{Properties of regular play}\strut 
\indent The next two claims describe the state of the different gadgets throughout regular play.

\begin{claim}\label{cla:regular}
    Consider the situation at the end of a sequence of regular play, assuming that both players have followed regular play from the start. Then, Maker has not yet filled an edge. Moreover, for all $t \in N$, denoting by $H_t'$ the updated version of the gadget $H_t$:
    \begin{enumerate}[label={\textup{(\arabic*)}}]
        \item If $t$ has never been active, then the gadget is intact: $H_t'=H_t$. \label{item1}
        \item If $t$ is not active but has been active before, and $t \not\in M_{2,1}$, then $E(H_t')=\varnothing$. \label{item2}
        \item If $t$ is not active but has been active before, and $t \in M_{2,1}$, then, writing $g=\overrightarrow{rt} \in A$ where $r$ was not active immediately prior to $t$, we have $E(H_t')=\{\{p_g,q_g,z_g\}\}$. \label{item3}
        \item If $t$ is active for the first time, then, writing $g=\overrightarrow{rt} \in A$ where $r$ was active immediately prior to $t$, all vertices of $H_t$ are unpicked apart from $p_g$ and $q_g$ which have been picked by Maker. \label{item4}
        \item If $t$ is active for the second time and $t \in B_{2,1}$, then $E(H_t')=\varnothing$. \label{item5}
        \item If $t$ is active for the second time and $t \in M_{2,1}$, then, writing $g=\overrightarrow{rt} \in A$ where $r$ was active immediately prior to $t$, we have $E(H_t')=\{\{z_g\}\}$. \label{item6}
    \end{enumerate}
\end{claim}

\begin{proof}
    We prove the claim by induction on the number of sequences of regular play that have been played.
    \begin{itemize}
        \item First of all, consider the situation at the end of the first sequence of regular play, which takes place inside the gadget $H_s$. The node that was active before that sequence is $s \in B_{0,1}$, and the active node is $w$ where $w$ denotes the only out-neighbor of $s$. Let $a=\overrightarrow{sw} \in A$. Breaker has picked $y_1$ and $y_2$, hence killing all edges of $H_s$, and leaving all other gadgets intact since $y_1$ and $y_2$ are interior vertices of $H_s$. Meanwhile, Maker has picked $p_a$ and $q_a$, which are input vertices of the gadget $H_w$. Therefore, $s$ satisfies item \ref{item2} of the claim and $w$ satisfies item \ref{item4} of the claim, while all nodes other than $s$ and $w$ have intact updated gadgets and thus satisfy item \ref{item1} of the claim. Finally, Maker has clearly not filled an edge yet.
        \item Now, consider the situation at the end of a sequence $\sequence$ of regular play which was not the first one. Suppose that the claim was satisfied before the sequence $\sequence$: we want to show that it still is. There exists $e=\overrightarrow{vw} \in A$ such that $w$ is the active node and $v$ was active immediately prior. Since $\sequence$ only involved interior vertices of $H_v$ as well as the junction vertices $p_e$ and $q_e$, it has left unchanged all updated gadgets $H'_t$ such that $t \not\in \{v,w\}$: by the induction hypothesis, those nodes satisfy the claim. It remains to be shown that $v$ and $w$ also do, and that Maker has not yet filled an edge.
        \begin{itemize}
            \item[--] Let us first show that $v$ satisfies the claim. Note that $v \neq s$ since $\sequence$ is not the first sequence of regular play. In particular, we can define $f=\overrightarrow{uv} \in A$ where $u$ was active immediately prior to $v$. Recall that, by item \ref{item4} of the induction hypothesis applied to $v$, before the sequence $\sequence$, $p_f$ and $q_f$ had already been picked by Maker and all other vertices of $H_v$ were unpicked.
            
            If $v \not\in M_{2,1}$ then it is easy to see that, in all cases, the sequence $\sequence$ has Breaker killing all edges of $H_v$, i.e. the set of vertices picked by Breaker during the sequence $\sequence$ intersects every edge of $H_v$. Therefore, $v$ satisfies item \ref{item2} of the claim.
            
            If $v \in M_{2,1}$, then, writing $g=\overrightarrow{rv} \in A$ where $r$ was not active immediately prior to $v$, we see that the sequence $\sequence$ has the players picking all remaining vertices of $H_v$ apart from $p_g$, $q_g$ and $z_g$ and has Breaker killing all edges of $H_v$ apart from the edge $\{p_g,q_g,x_1,z_g\}$. Since Maker has picked $x_1$, this yields the edge $\{p_g,q_g,z_g\}$ in the updated hypergraph, so $v$ satisfies item \ref{item3} of the claim.
            \item[--] We have just seen that Breaker's picks during the sequence $\sequence$ kill all edges of $H_v$, except for a single edge which is not filled by Maker anyway when $v \in M_{2,1}$. The only vertices involved in the sequence $\sequence$ that could have impacted other gadgets are $p_e$ and $q_e$, which have been picked by Maker, but the induction hypothesis ensures that the updated hypergraph before the sequence $\sequence$ did not have $\{p_e,q_e\}$ (or any subset of $\{p_e,q_e\}$) as an edge. All in all, Maker has not yet filled an edge.
            \item[--] Now, let us show that $w$ satisfies the claim.
            If $w$ is active for the first time, then item \ref{item1} of the induction hypothesis applied to $w$ ensures that its gadget was intact before the sequence $\sequence$, and the only difference after the sequence $\sequence$ is that $p_e$ and $q_e$ have been picked by Maker: this means that $w$ satisfies item \ref{item4} of the claim.
            If $w$ is active for the second time and $w \in B_{2,1}$, then item \ref{item2} of the induction hypothesis applied to $w$ ensures that its updated gadget before the sequence $\sequence$ already had no edge, so $w$ satisfies item \ref{item5} of the claim.
            Finally, if $w$ is active for the second time and $w \in M_{2,1}$, then item \ref{item3} of the induction hypothesis applied to $w$ ensures that its updated gadget before the sequence $\sequence$ had $\{p_e,q_e,z_e\}$ as its only edge: since $p_e$ and $q_e$ have now been picked by Maker, the current updated gadget has $\{z_e\}$ as its only edge, so $w$ satisfies item \ref{item6} of the claim. This ends the proof. \qedhere
        \end{itemize}
    \end{itemize}
\end{proof}

\begin{claim}\label{cla:pairing2}
    Consider the situation at the end of a sequence of regular play, assuming that both players have followed regular play from the start. For every $t \in N$ which is not the active node, the updated gadget $H'_t$ admits all the following pairings:
    \begin{itemize}[noitemsep,nolistsep]
        \item a clean pairing $\Pi'(t)$;
        \item for each vertex $z \in V(H_t')$ which is not an input vertex of $H_t$: a clean pairing $\Pi'(t,z)$ which does not use $z$;
        \item for each vertex $z \in V(H_t')$ which is an input vertex of $H_t$: a pairing $\Pi'(t,z)$ which does not use $z$ and whose only mixed pair, if any, contains the twin of $z$ paired with an interior vertex.
    \end{itemize}
\end{claim}

\begin{proof}
    If $t$ has never been active, then $H'_t=H_t$ by item \ref{item1} of Claim \ref{cla:regular}. Therefore, we can simply use the pairings given by Claim \ref{cla:pairing}: we define $\Pi'(t) = \Pi(t)$ and $\Pi'(t,z)=\Pi(t,z)$ for all $z \in V(H_t)$. If $t$ has been active before and $t \not\in M_{2,1}$, then $E(H'_t)=\varnothing$ by item \ref{item2} of Claim \ref{cla:regular}, so we can define $\Pi'(t) = \varnothing$ and $\Pi'(t,z)=\varnothing$ for all $z \in V(H'_t)$, which trivially are clean pairings. Finally, if $t$ has been active before and $t \in M_{2,1}$, then item \ref{item3} of Claim \ref{cla:regular} ensures that, writing $g=\overrightarrow{rt} \in A$ where $r$ was not active immediately prior to $t$, we have $E(H_t')=\{\{p_g,q_g,z_g\}\}$. Therefore, we can define $\Pi'(t)=\Pi'(t,z_g)=\{\{p_g,q_g\}\}$, $\Pi'(t,p_g)=\{\{q_g,z_g\}\}$ (the pair is mixed, but it is allowed since $p_g$ is an input vertex of $H_t$ and $z_g$ is an interior vertex) and $\Pi'(t,q_g)=\{\{p_g,z_g\}\}$.
\end{proof}

We now show that regular play always leads to the desired result.

\begin{claim}\label{cla:result}
    Assume that both players are constrained to follow regular play until it ends.
    \begin{itemize}
        \item If Alice has a winning strategy for the Geography game on $(N,A,s)$, then Maker has a winning strategy for the Maker-Breaker game on $H$ where she wins on the first move succeeding regular play.
        \item If Bob has a winning strategy for the Geography game on $(N,A,s)$, then Breaker has a winning strategy for the Maker-Breaker game on $H$ thanks to a (clean) pairing of the updated hypergraph obtained at the end of regular play.
    \end{itemize}
\end{claim}

\begin{proof}
    We shall rename Maker as Alice and Breaker as Bob for this proof. First, suppose that Alice has a winning strategy $\strat$ for the Geography game on $(N,A,s)$. While playing the Maker-Breaker game on $H$, Alice plays a fictitious Geography game on $(N,A,s)$, and goes back and forth between the two games as follows, in a way that ensures that the node hosting the token in the Geography game is always the same as the active node of regular play in the Maker-Breaker game. First, she places the token on $s$ in the Geography game, as she must. After that:
    \begin{itemize}
        \item[--] When it is Bob's turn in the Geography game (i.e. the token is placed on a node $v$ with same color as $s$), Alice initiates the sequence of regular play inside the gadget $H_v$ in the Maker-Breaker game. Note that, since $v \in B_{0,1} \cup B_{1,1} \cup B_{2,1} \cup B_{1,2}$, Alice has no choice to make during that sequence. After the sequence is over, she takes note of the new active node $w$, and considers that Bob has moved the token from $v$ to $w$ in the Geography game.
        \item[--] When it is Alice's turn in the Geography game (i.e. the token is placed on a node $v$ with different color than $s$), Alice moves the token to the node $w$ suggested by the strategy $\strat$ in the Geography game. She then initiates the sequence of regular play inside the gadget $H_v$ in the Maker-Breaker game, so that the new active node is $w$. Note that, since $v \in M_{1,1} \cup M_{2,1} \cup M_{1,2}$, Alice can indeed choose which of the two out-neighbors of $v$ to activate if $v$ has out-degree 2.
    \end{itemize}
    Since $\strat$ is a winning strategy, Bob will at some point lose the Geography game by moving the token to a node $w$ (necessarily of different color than $s$) that had already been visited. This means that, in the Maker-Breaker game, the node $w$ has been activated for a second time, thus ending the phase of regular play. Since $w$ has been visited twice, it must have in-degree 2, so $w \in M_{2,1}$. Therefore, according to item \ref{item6} of Claim \ref{cla:regular}, the updated hypergraph contains an edge of size 1 at this point, so Alice can win in one move by picking that vertex.
    \\ Now, suppose that Bob has a winning strategy $\strat$ for the Geography game on $(N,A,s)$. We use an analogous reasoning to the above, with Bob converting Alice's choices in the Maker-Breaker game into moves in the Geography game, while echoing his own moves in the Geography game (suggested by $\strat$) into the Maker-Breaker game. Since $\strat$ is a winning strategy, Alice will at some point lose the Geography game by moving the token to a node $w$ (necessarily of same color as $s$) that had already been visited. This means that, in the Maker-Breaker game, the node $w$ has been activated for a second time, thus ending the phase of regular play. Since $w \in B_{2,1}$, we know by item \ref{item5} of Claim \ref{cla:regular} that the updated gadget of $w$ has no edge. For every other updated gadget, Claim \ref{cla:pairing2} provides a clean pairing $\Pi'(\cdot)$, so the full updated hypergraph admits the following clean pairing:
    $$ \mathbf{\Pi} = \bigcup_{t \in V \setminus \{w\}} \Pi'(t).$$
    Therefore, Bob has a winning strategy from this point on by Lemma \ref{lem:pairing}.
\end{proof}

\subsection{Optimality of regular play}\strut 
\indent We now verify that regular play is optimal for both players.

\begin{claim}\label{cla:optimalgreedy}
    Assume that both players have followed regular play thus far, and that regular play is not over. Suppose it is Maker's turn, and regular play suggests a move labelled ``g'' (resp. ``2-g'', resp. ``3-g'') in Tables \ref{tab:tableau1}--\ref{tab:tableau2}--\ref{tab:tableau3}. Then, that move is indeed a 1-greedy move (resp. 2-greedy move, resp. 3-greedy move) as per Definition \ref{def:greedy}. 
    More specifically: any Breaker answer suggested (resp. not suggested) by regular play falls under item (ii) (resp. item (i)) of Definition \ref{def:greedy}.
\end{claim}

\begin{proof}
	We detail what happens for the gadgets $H_s$ and $H_v$ when $v \in B_{1,2} \cup B_{2,1}$. These three cases encapsulate all types of situations that arise across all gadgets.
	
	\begin{itemize}
	
	\item Let us start with $H_s$ (recall the figure from Table \ref{tab:tableau1}). Regular play suggests that Maker starts the game by picking $p_a$ (labelled ``3-g'') and that Breaker responds by picking $y_1$. Clearly, every edge of $H$ containing $y_1$ also contains $p_a$. Therefore, it suffices to show that, if Maker picks $p_a$ and Breaker answers by picking some $y \neq y_1$, then Maker wins in three more moves. This is due to the fact that, after Breaker picks $y$, Maker can pick $y_1$ herself. Breaker then picks some $y'$. Between three and five of the vertices $z_1,z_2,z_3,z_4,z_5$ are unpicked at this point (depending on whether $y$ and $y'$ are among them or not), so Maker can pick two arbitrary vertices among $z_1,z_2,z_3,z_4,z_5$ with her next two moves, thus filling an edge. The same argument works for Maker's second move of the game ($q_a$, also labelled ``3-g'') and Breaker's answer $y_2$.
	
	\item Let us now consider an active node $v \in B_{1,2}$ (recall the figure from Table \ref{tab:tableau2}). Before the sequence of regular play on $H_v$ begins, Claim \ref{cla:regular} ensures that $p_a$ and $q_a$ have been picked by Maker while all other vertices of $H_v$ are unpicked. We go through the sequence of regular play in order.
    \begin{itemize}
        \item[--] Before the sequence starts, there is an updated edge $\{x_1,y_1\}$ (yielded by the original edge $\{p_a,q_a,x_1,y_1\}$). Therefore, if Maker picks $x_1$ and Breaker does not pick $y_1$ as an answer, then Maker wins in one more move by picking $y_1$ herself. Moreover, that updated edge $\{x_1,y_1\}$ is the only one containing $y_1$, and it contains $x_1$ as well. All in all, $x_1$ is indeed a 1-greedy move, which forces Breaker to pick $y_1$.
        \item[--] The case of the 1-greedy move $x_2$ and Breaker's answer $y_2$ is analogous. Note that Maker could actually have made the 1-greedy moves $x_1$ and $x_2$ is any order, as both were available at the start of the sequence.
        \item[--] After these four moves, regular play suggests that Maker picks $x_3$ and Breaker picks $y_3$ or $y_4$. We skip this part for now since this case is not covered by the present claim ($x_3$ is not a greedy move). By symmetry, assume that Breaker has picked $y_4$.
        \item[--] The original edge $\{p_b,x_1,x_3,y_3\}$ has now yielded the edge $\{p_b,y_3\}$ in the updated hypergraph since Maker has picked $x_1$ and $x_3$. Moreover, this is the only updated edge containing $y_3$: indeed, Breaker killed the edge $\{x_1,x_2,y_3,y_4\}$ by picking $y_4$ previously. Therefore, $p_b$ is a 1-greedy move, which forces Breaker to pick $y_3$.
        \item[--] The final two moves of the sequence are not covered by the present claim as $q_b$ is not a greedy move.
    \end{itemize}
    
    \item Finally, let us consider an active node $v \in B_{2,1}$ (recall the figure from Table \ref{tab:tableau3}). The 1-greedy moves are straightforward using analogous arguments to the previous case. We now explain the 2-greedy move $x_2$, which happens after four moves have been made in the current sequence of regular play. Clearly, all edges of $H$ containing $z_1$, $z_2$ or $z_3$ also contain $x_2$. Therefore, it suffices to show that, if Maker picks $x_2$ and Breaker picks some $y \not\in \{z_1,z_2,z_3\}$, then Maker can win in two more moves. This comes from the fact that, after Maker picks $x_2$, there are updated edges $\{z_1,z_2\}$ and $\{z_2,z_3\}$ (yielded by the original edges $\{x_1,x_2,z_1,z_2\}$ and $\{x_1,x_2,z_2,z_3\}$, in which Maker has picked $x_1$ and $x_2$): if Breaker picks $y \not\in \{z_1,z_2,z_3\}$, then Maker can pick $z_2$ herself and win on her next move by picking $z_1$ or $z_3$. \qedhere
	\end{itemize}
\end{proof}

\begin{claim}\label{cla:optimalbreaker}
	Assume that both players have followed regular play thus far, and that regular play is not over. If it is Breaker's turn, then it is optimal for Breaker to follow regular play with his next move. More specifically: Breaker violating regular play during the very first sequence of regular play (which takes place inside $H_s$) allows Maker to fill an edge in three more moves, while Breaker violating regular play later on allows Maker to fill an edge in two more moves.
\end{claim}

\begin{proof}

We distinguish three types of moves suggested by regular play for Breaker.

\begin{itemize}
	
	\item First, there are the moves that come in response to a greedy move by Maker. Those are already covered by Claim \ref{cla:optimalgreedy}.
	
	\item Second, there are some moves labelled ``f'' that do not come in response to a greedy move by Maker. It is easy to check that these are indeed forced moves. For an active node in $B_{1,2}$, when Breaker is supposed to pick $y_5$ as the final move of that sequence, that move is forced since Maker has already picked all the other vertices of the edge $\{p_b,q_b,x_3,y_5\}$ or $\{p_c,q_c,x_3,y_5\}$ depending on the scenario. For an active node in $M_{1,2}$, when Breaker is supposed to pick $z_1$ (resp. $z_2$) after Maker has picked $z_2$ (resp. $z_1$), that move is forced since Maker has already picked all the other vertices of the edge $\{x_1,x_2,z_1,z_2\}$.
	
	\item Finally, there is the case of an active node $v \in B_{1,2}$, as Breaker chooses between picking $y_3$ or $y_4$. This choice happens after five moves have been made in the current sequence, with the updated gadget being pictured in Figure \ref{fig:breakerchoice}.
  
\begin{figure}[h]
\centering
\includegraphics[scale=.55]{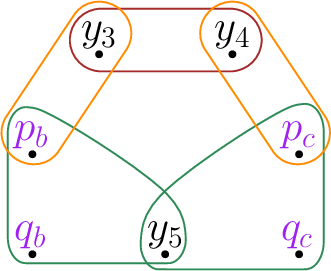}
\caption{The gadget $H_v$ for $v \in B_{1,2}$, as updated after the first five moves of the sequence of regular play on $H_v$.}\label{fig:breakerchoice}
\end{figure}

We can see that Breaker must pick one of $y_3$, $y_4$ or $p_b$, otherwise Maker can win in two moves by picking $y_3$ and then $y_4$ or $p_b$. Similarly, Breaker must pick one of $y_3$, $y_4$ or $p_c$, otherwise Maker can win in two moves by picking $y_4$ and then $y_3$ or $p_b$. These two facts combined ensure that Breaker must indeed pick either $y_3$ or $y_4$ to avoid losing after two more Maker moves. \qedhere
\end{itemize}
\end{proof}

\begin{claim}\label{cla:optimalmaker}
    Assume that both players have followed regular play thus far, and that regular play is not over. If it is Maker's turn, then it is optimal for Maker to follow regular play with her next move. More specifically: in all situations where regular play suggests a move labelled ``g'' (resp. ``2-g'', resp. ``3-g'') in Tables \ref{tab:tableau1}--\ref{tab:tableau2}--\ref{tab:tableau3}, then that move is indeed a 1-greedy move (resp. 2-greedy move, resp. 3-greedy move) as per Definition \ref{def:greedy}, whereas in all other situations, violating regular play yields a losing position for Maker.
\end{claim}

\begin{proof}
	Claim \ref{cla:optimalgreedy} already covers all greedy moves, which we know to be optimal by Lemma \ref{lem:greedy}. Therefore, we only need to address Maker's moves which are not labelled ``g'', ``2-g'' or ``3-g''. 
    Looking at Tables \ref{tab:tableau1}--\ref{tab:tableau2}--\ref{tab:tableau3}, we can see there are three such situations. Case 1 is that of an active node $v \in M_{1,2}$, as Maker chooses between picking $z_1$ or $z_2$. Case 2 is that of an active node $v \in B_{1,2}$, when Maker picks $x_3$. Case 3 is that of an active node $v \in B_{1,2}$, when Maker picks $q_b$ or $q_c$. We show that, if Maker picks an irregular vertex, then Breaker has an answer after which the full updated hypergraph admits a pairing $\mathbf{\Pi}$, which concludes according to Lemma \ref{lem:pairing}. In each of these three cases, the pairing $\mathbf{\Pi}$ will be constructed using a union of the pairings $\Pi'(t)$ and $\Pi'(t,\cdot\,)$, given for each non-active node $t$ by Claim \ref{cla:pairing2} (applied to the situation obtained at the end of the previous sequence of regular play).
    
    \begin{enumerate}[label={(\arabic*)}]
    
        \item Case 1 is that of an active node $v \in M_{1,2}$, as Maker must choose between picking $z_1$ or $z_2$. The updated gadget of $v$ at this point is pictured in Figure \ref{fig:makerchoice1}. Suppose that Maker picks some $x \not\in \{z_1,z_2\}$ instead: we now show that this is a losing move. Let $w_b$ and $w_c$ be the out-neighbors of $v$ through the arcs $b$ and $c$ respectively.
        \begin{figure}[h]
        \centering
        \includegraphics[scale=.55]{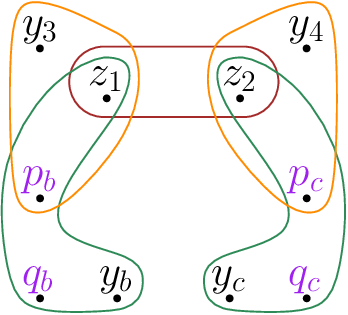}
        \caption{The gadget $H_v$ for $v \in M_{1,2}$, as updated after the first four moves of the sequence of regular play on $H_v$.}\label{fig:makerchoice1}
        \end{figure}
        \begin{enumerate}[label={(\alph*)}]
            \item First, suppose that $x \in \{y_3,y_4,y_b,y_c\}$. That move does not impact any other gadget since $x$ is an interior vertex. Now, Breaker picks $z_1$. Note that $\{\{p_c,z_2\}\}$ is a pairing for (the updated gadget of) $v$ after that move. We should not use $p_c$ in the pairing for $w_c$ since we already use it for $v$. We thus conclude by defining the pairing $\mathbf{\Pi}$ as:
            $$ \mathbf{\Pi} = \{\{p_c,z_2\}\} \cup \Pi'(w_c,p_c) \cup \bigcup_{t \in N \setminus \{v,w_c\}}\Pi'(t).$$
            \item Next, suppose that $x \in \{p_b,q_b,p_c,q_c\}$. By symmetry, assume that $x \in \{p_b,q_b\}$. Now, Breaker picks $z_1$. Note that $\{\{p_c,z_2\}\}$ is a pairing for $v$ after that move.  We should not use $p_c$ in the pairing for $w_c$ since we already use it for $v$, nor should we use $x$ in the pairing for $w_b$ since it has already been picked by Maker. We thus define:
            $$ \mathbf{\Pi} = \{\{p_c,z_2\}\} \cup \Pi'(w_b,x) \cup \Pi'(w_c,p_c) \cup \bigcup_{t \in N \setminus \{v,w_b,w_c\}}\Pi'(t).$$
            \item To end the proof of Case 1, suppose that $x \not\in V(H_v)$.

            \begin{itemize}
            	\item First, suppose that $x$ is an an interior vertex of some (necessarily unique) gadget $H_u$. Now, Breaker picks $z_1$. Note that $\{\{p_c,q_c\},\{y_4,z_2\}\}$ is a clean pairing for $v$ after that move. We thus define:
                $$ \mathbf{\Pi} = \{\{p_c,q_c\},\{y_4,z_2\}\} \cup \Pi'(u,x) \cup \bigcup_{t \in N \setminus \{v,u\}}\Pi'(t).$$
                \item Second, suppose that $x$ is a junction vertex i.e. $x \in \{p_e,q_e\}$ for some $e=\overrightarrow{u_1u_2} \in A$. Since the digraph $(N,A)$ is bipartite, we know that $\{w_b,w_c\} \neq \{u_1,u_2\}$. By symmetry, assume $w_c \not\in \{u_1,u_2\}$. Now, Breaker picks $z_1$. Note that $\{\{p_c,z_2\}\}$ is a pairing for $v$ after that move. We should not use $p_c$ in the pairing for $w_c$ since we already use it for $v$. We thus define:
            $$ \mathbf{\Pi} = \{\{p_c,z_2\}\} \cup \Pi'(w_c,p_c) \cup \Pi'(u_1,x) \cup \Pi'(u_2,x) \cup \bigcup_{t \in N \setminus \{v,w_c,u_1,u_2\}}\Pi'(t).$$
            \end{itemize}
        \end{enumerate}
        
        \item Case 2 is that of an active node $v \in B_{1,2}$, just before Maker picks $x_3$. The updated gadget of $v$ at this point is pictured in Figure \ref{fig:makerchoice2}. Suppose that Maker picks some $x \neq x_3$ instead: we now show that this is a losing move. Let $w_b$ and $w_c$ be the out-neighbors of $v$ through the arcs $b$ and $c$ respectively.
        \begin{figure}[h]
        \centering
        \includegraphics[scale=.55]{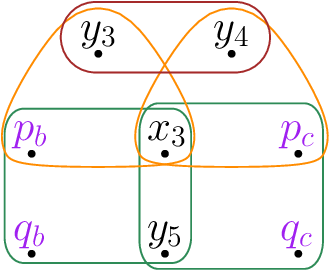}
        \caption{The gadget $H_v$ for $v \in B_{1,2}$, as updated after the first four moves of the sequence of regular play on $H_v$.}\label{fig:makerchoice2}
        \end{figure}
        \begin{enumerate}[label={(\alph*)}]
            \item First, suppose that $x \in \{y_3,y_4\}$. By symmetry, assume that $x=y_3$. Now, Breaker picks $y_4$ (which is forced anyway). Note that $\{\{p_b,x_3\},\{p_c,q_c\}\}$ is a pairing for $v$ after that move. We should not use $p_b$ in our pairing for $w_b$, so we define:
            $$ \mathbf{\Pi} = \{\{p_b,x_3\},\{p_c,q_c\}\} \cup \Pi'(w_b,p_b) \cup \bigcup_{t \in N \setminus \{v,w_b\}}\Pi'(t).$$
            \item Next, suppose that $x \in \{p_b,q_b,p_c,q_c,y_5\}$. By symmetry, assume that $x \in \{p_b,q_b,y_5\}$. Now, Breaker picks $x_3$. Note that $\{\{y_3,y_4\}\}$ is a clean pairing for $v$ after that move.
            \begin{itemize}
                \item If $x \in \{p_b,q_b\}$, then we define:
                $$ \mathbf{\Pi} = \{\{y_3,y_4\}\} \cup \Pi'(w_b,x) \cup \bigcup_{t \in N \setminus \{v,w_b\}}\Pi'(t).$$
                \item If $x = y_5$, then $x$ is an interior vertex of $H_v$, so we can simply define:
                $$ \mathbf{\Pi} = \{\{y_3,y_4\}\} \cup \bigcup_{t \in N \setminus \{v\}}\Pi'(t).$$
            \end{itemize}
            \item To end the proof of Case 2, suppose that $x \not\in V(H_v)$. Now, Breaker picks $x_3$. Note that $\{\{y_3,y_4\}\}$ is a clean pairing for $v$ after that move. 
            \begin{itemize}
            	\item If $x$ is an interior vertex of some (necessarily unique) gadget $H_u$, then we define:
                $$ \mathbf{\Pi} = \{\{y_3,y_4\}\} \cup \Pi'(u,x) \cup \bigcup_{t \in N \setminus \{v,u\}}\Pi'(t).$$
                \item If $x$ is a junction vertex i.e. $x \in \{p_e,q_e\}$ for some $e=\overrightarrow{u_1u_2} \in A$, then we define:
                $$ \mathbf{\Pi} = \{\{y_3,y_4\}\} \cup \Pi'(u_1,x) \cup \Pi'(u_2,x) \cup \bigcup_{t \in N \setminus \{v,u_1,u_2\}}\Pi'(t).$$
                
            \end{itemize}
        \end{enumerate}
        
        \item Case 3 is that of an active node $v \in B_{1,2}$, just before Maker picks $q_b$ (by symmetry, we assume that Breaker has chosen $y_4$ rather than $y_3$ earlier during that sequence). The updated gadget of $v$ at this point is pictured in Figure \ref{fig:makerchoice3}. Suppose that Maker picks some $x \neq q_b$ instead: we now show that this is a losing move. Let $w_b$ and $w_c$ be the out-neighbors of $v$ through the arcs $b$ and $c$ respectively.
        \begin{figure}[h]
        \centering
        \includegraphics[scale=.55]{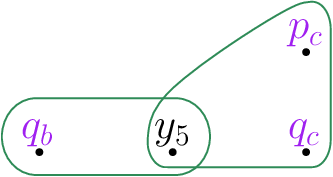}
        \caption{The gadget $H_v$ for $B \in M_{1,2}$, as updated after the first eight moves of the sequence of regular play on $H_v$.}\label{fig:makerchoice3}
        \end{figure}
        \begin{enumerate}[label={(\alph*)}]
            \item First, suppose that $x=y_5$. Now, Breaker picks $q_b$ (which is forced anyway). Note that $\{\{p_c,q_c\}\}$ is a clean pairing for $v$ after that move. Moreover, we can see that $\Pi'(w_b) \setminus \{\{p_b,q_b\}\}$ is a clean pairing for $w_b$ at that point. Indeed, Maker has picked $p_b$ but Breaker has picked $q_b$, and every edge of $H_{w_b}$ containing $p_b$ also contains $q_b$ (this is true for every input pair of every gadget). All in all, we define:
            $$ \mathbf{\Pi} = \{\{p_c,q_c\}\} \cup (\Pi'(w_b) \setminus \{\{p_b,q_b\}\}) \cup \bigcup_{t \in N \setminus \{v,w_b\}}\Pi'(t).$$
            \item Next, suppose that $x \in \{p_c,q_c\}$. Now, Breaker picks $y_5$. Note that $\varnothing$ is a clean pairing for $v$ after that move, as Breaker has killed all edges in $H_v$. Therefore, we define: $$ \mathbf{\Pi} = \Pi'(w_c,x) \cup \bigcup_{t \in N \setminus \{v,w_c\}}\Pi'(t).$$
            \item To end the proof of Case 3, suppose that $x \not\in V(H_v)$. Now, Breaker picks $y_5$. Note that $\varnothing$ is a clean pairing for $v$ after that move, as Breaker has killed all edges in $H_v$. The reasoning is now analogous to Case 2(c):
            \begin{itemize}
            	\item If $x$ is an interior vertex of some (necessarily unique) gadget $H_u$, then we define:
                $$ \mathbf{\Pi} = \Pi'(u,x) \cup \bigcup_{t \in N \setminus \{v,u\}}\Pi'(t).$$
                \item If $x$ is a junction vertex i.e. $x \in \{p_e,q_e\}$ for some $e=\overrightarrow{u_1u_2} \in A$, then we define:
                $$ \mathbf{\Pi} = \Pi'(u_1,x) \cup \Pi'(u_2,x) \cup \bigcup_{t \in N \setminus \{v,u_1,u_2\}}\Pi'(t).$$
            \end{itemize}
        \end{enumerate}
    \end{enumerate}
    This ends the proof of the claim, as all cases have been addressed.
\end{proof}

\subsection{Concluding statements}

Putting Claims \ref{cla:result}, \ref{cla:optimalbreaker} and \ref{cla:optimalmaker} together, we can see that Maker has a winning strategy for the Maker-Breaker game on $H$ if and only if Alice has a winning strategy for the Geography game on $(N,A,s)$. Note that $|V(H)|=O(|N|)$, $|E(H)|=O(|N|)$, and $H$ can be constructed from $(N,A,s)$ in linear time. This achieves the desired reduction and ends the proof of Theorem \ref{the:mainMB}.

\begin{remark}\label{rem:boundeddegree}
	We have actually shown that the Maker-Breaker problem is {\sf PSPACE}-complete even when restricted to 4-uniform hypergraphs of bounded maximum degree. More specifically, it can easily be checked that all vertices in our construction have degree at most 5, except for the vertices $p_a,q_a,y_1,y_2$ in the gadget $H_s$ which are of degree between 10 and 12. Note that the higher degree of these four vertices is only due to the fact that we were aiming for a 4-uniform construction. If one simply wants a construction of rank 4, then the gadget $H_s$ can be replaced by two edges $\{p_a,y_1\}$ and $\{q_a,y_2\}$, which yields an overall construction of maximum degree at most 5. This improves on the known {\sf PSPACE}-completeness result for the Maker-Breaker game on hypergraphs of rank 12 and maximum degree 5 \cite{boundeddegreeQBF}.
\end{remark}

\section{{\PSPACE}-completeness of 4-uniform Maker-Maker games}\label{section4}\strut
\indent In this section, we are going to show the following result.

\begin{theorem}\label{the:mainMM}
    Deciding whether FP has a winning strategy for the Maker-Maker game on a 4-uniform hypergraph is a {\sf PSPACE}-complete problem.
\end{theorem}

Since membership in {\sf PSPACE} is well documented for Maker-Maker games \cite{Byskov}, it suffices to show {\sf PSPACE}-hardness. For this, we show that our construction $H$ from Section \ref{section3} is such that Maker has a winning strategy for the Maker-Breaker game on $H$ if and only if FP has a winning strategy for the Maker-Maker game on $H$. 

We define regular play the exact same as in Section \ref{section3}, with FP in the role of Maker and SP in the role of Breaker. We must show that regular play is still optimal in the Maker-Maker convention, and that it leads to the desired result: FP winning if Alice wins the Geography game, or a draw if Bob wins the Geography game.

As the players follow regular play, the updated red edges in the Maker-Maker convention are the same as the updated edges were in the Maker-Breaker convention, so FP makes the same threats that Maker would in the Maker-Breaker convention. However, we need to check that the updated blue edges are not an issue, i.e. that SP never threatens anything of his own. 
The next claim is key in that regard.

\begin{claim}\label{cla:nothreat}
At any point throughout regular play and assuming both players follow it, all updated blue edges are of size 3 or 4, and those of size 3 are precisely of the form $\{p_b,q_b,y_b\}$ or $\{p_b,q_b,y_c\}$ inside gadgets of nodes in $M_{1,2}$ with the notations from Table \ref{tab:tableau2}. Moreover, throughout the first sequence of regular play, all updated blue edges are of size 4.
\end{claim}

\begin{proof}
	This comes from the following simple observation: almost every time regular play suggests some move $y$ for SP, all edges containing $y$ have already been killed by FP at that point, meaning that they contain at least one vertex that has been picked by FP earlier during regular play. The only exception happens for an active node in $M_{1,2}$, when regular play suggests that SP picks $y=z_1$ (resp. $y=z_2$), as there is a single updated blue edge $\{p_b,q_b,y_b,z_1\}$ (resp. $\{p_c,q_c,y_c,z_2\}$) containing $y$ at that point, which yields the updated blue edge $\{p_b,q_b,y_b\}$ (resp. $\{p_c,q_c,y_c\}$).
	
    Let us detail one case, which is that of an active node in $M_{1,2}$, with the notations from Table \ref{tab:tableau2}. FP starts the sequence by picking $x_1$. SP now picks $y_1$: the only edge of $H$ containing $y_1$ is the edge $\{p_a,q_a,x_1,y_1\}$, but that edge has already been killed by FP as she picked $x_1$ (actually, FP had killed it even earlier by picking $p_a$ and $q_a$). The same thing happens when FP picks $x_2$ and SP picks $y_2$. Say FP picks $z_1$ next. SP now picks $z_2$: the only edges of $H$ containing $z_2$ are $\{x_1,x_2,z_1,z_2\}$, which has been killed by FP as she picked $x_1$, and $\{p_c,q_c,y_c,z_2\}$, which yields the blue edge $\{p_c,q_c,y_c\}$ as it has not been killed by FP. Now, FP picks $p_b$ and SP picks $y_3$: the only edge of $H$ containing $y_3$ is $\{p_b,x_1,y_3,z_1\}$, which has been killed by FP as she picked $x_1$. Finally, FP picks $q_b$, and SP picks $y_b$: the only edge of $H$ containing $y_b$ is $\{p_b,q_b,y_b,z_1\}$, which has been killed by FP as she picked $z_1$.
\end{proof}

The next claim ensures that regular play always leads to the desired result for the Maker-Maker game.

\begin{claim}\label{cla:resultMM}
    Assume that both players are constrained to follow regular play until it ends.
    \begin{itemize}
        \item If Alice has a winning strategy for the Geography game on $(N,A,s)$, then FP has a winning strategy for the Maker-Maker game on $H$ where she wins on the first move succeeding regular play.
        \item If Bob has a winning strategy for the Geography game on $(N,A,s)$, then both players have a drawing strategy for the Maker-Maker game on $H$.
    \end{itemize}
\end{claim}

\begin{proof}
    For the most part, the proof is the same as for the Maker-Breaker convention (Claim \ref{cla:result}), since regular play is defined the same and also comes to its term without a winner (indeed, we know FP does not win by Claim \ref{cla:regular}, and SP does not win by Claim \ref{cla:nothreat}). If Alice has a winning strategy for the Geography game on $(N,A,s)$, then FP has a winning strategy for the Maker-Maker game on $H$ where she wins on the first move succeeding regular play, as she does in the Maker-Breaker convention by Claim \ref{cla:result}. If Bob has a winning strategy for the Geography game on $(N,A,s)$, then SP has a drawing strategy for the Maker-Maker game on $H$, using the clean pairing $\mathbf{\Pi}$ of the updated red edges given by Claim \ref{cla:result}. It only remains to be shown that, in the case where Bob is the one that has a winning strategy for the Geography game on $(N,A,s)$, FP also has a drawing strategy at the end of regular play. This is not a given: for all we know so far, FP might be losing at the end of regular play because it had her play suboptimal moves.
    
    Using Lemma \ref{lem:pairing}, it suffices to show that the updated blue edges at the end of regular play admit a pairing. Recall that, by Claim \ref{cla:nothreat}, all such edges are either of size 4 or of the form $\{p_e,q_e,y_e\}$. Those of size 4 are original edges of $H$, so they exist in both red and blue and are taken care of by the same clean pairing $\mathbf{\Pi}$ that SP uses to draw the game. As for each one of the form $\{p_e,q_e,y_e\}$, it suffices to add the clean junction pair $\{p_e,q_e\}$ to the clean pairing $\mathbf{\Pi}$ (if it did not contain it already). All in all, we get a pairing of the updated blue edges, which ends the proof.
\end{proof}

Finally, we show that regular play is the optimal route for both players.

\begin{claim}\label{cla:optimalSP}
    Assume that both players have followed regular play thus far, and that regular play is not over. If it is SP's turn, then it is optimal for SP to follow regular play with his next move, as any other move puts him in a losing position.
\end{claim}

\begin{proof}
	Suppose that SP is first to violate regular play, and does so during the first sequence (resp. during the second sequence or later). In the Maker-Breaker convention, we know by Claim \ref{cla:optimalbreaker} that any such violation of regular play by Breaker allows Maker to fill an edge in three more moves (resp. in at most two more moves). The only way that SP could get away with it is if he could win faster than that himself: however, this is not the case since, by Claim \ref{cla:nothreat}, all updated blue edges before that violation by SP were of size 4 (resp. of size 3 or 4).
\end{proof}

\begin{claim}\label{cla:optimalFP}
    Assume that both players have followed regular play thus far, and that regular play is not over. If it is FP's turn, then it is optimal for FP to follow regular play with her next move.
\end{claim}

\begin{proof}

We divide FP's moves during regular play into two categories: greedy moves (labelled ``g'', ``2-g'' or ``3-g'') on one side, and the rest on the other.
	
\begin{itemize}

\item Suppose that regular play suggests a $k$-greedy move $x$ for some $k \in \{1,2,3\}$. We know by Claim \ref{cla:optimalgreedy} that $x$ is a $k$-greedy move in the Maker-Breaker sense, but we need to check that it is also a $k$-greedy move in the Maker-Maker sense as per Definition \ref{def:greedy2}. Let $y$ be a possible answer by SP in case FP picks $x$: we must show that $y$ satisfies item {\em (i)} or item {\em (ii)} from Definition \ref{def:greedy2}.

\begin{itemize}
	\item First, suppose that $y$ is not suggested by regular play as an answer to FP picking $x$. By Claim \ref{cla:optimalgreedy}, $y$ then falls under item {\em (i)} from Definition \ref{def:greedy} i.e. $y$ is a losing answer for Breaker in the Maker-Breaker convention, as it hands his opponent a guaranteed win in at most $k$ more moves. The same then holds for SP in the Maker-Maker convention: indeed, SP cannot win faster than that since, by Claim \ref{cla:nothreat}, the updated blue edges before SP picks $y$ are of size strictly larger than $k$ (recall that the case $k=3$ only occurs during the first sequence of regular play). Therefore, $y$ satisfies item {\em (i)} from Definition \ref{def:greedy2}.

	\item Second, suppose that $y$ is the move (or one of the moves, if there are several) suggested by regular play as an answer to FP picking $x$. By Claim \ref{cla:optimalgreedy}, $y$ then falls under item {\em (ii)} from Definition \ref{def:greedy} i.e. every updated red edge containing $y$ also contains $x$ (recall that updated edges in the Maker-Breaker convention correspond to updated red edges in the Maker-Maker convention). We now show that the same is true of the updated blue edges. Indeed, by Claim \ref{cla:nothreat}, all updated blue edges are of size 3 or 4, and those of size 3 are precisely of the form $\{p_b,q_b,y_b\}$ or $\{p_b,q_b,y_c\}$ inside gadgets of nodes in $M_{1,2}$ with the notations from Table \ref{tab:tableau2}. Updated blue edges of size 4 are original edges of $H$, so they exist in both red and blue, and we already know that every updated red edge containing $y$ also contains $x$. Finally, consider an updated blue edge of size 3, say $e=\{p_c,q_c,y_c\}$ inside $H_v$ for some $v \in M_{1,2}$ with the notations from Table \ref{tab:tableau2}. To show that $y \in e \implies x \in e$, it suffices to show that $y \not\in e$. Note that $e$ stems from SP picking $z_2$ earlier: in that scenario, regular play never suggests picking $y_c$, so $y \neq y_c$. Moreover, regular play never suggests that SP picks a junction vertex, so $y \neq p_c,q_c$. All in all, we have $y \not\in e$, which concludes: $y$ satisfies item {\em (ii)} from Definition \ref{def:greedy2}.
\end{itemize}
All in all, $x$ is indeed a greedy move, so it is optimal by Lemma \ref{lem:greedy2}.

\item Now, suppose that regular play does not suggest a greedy move. Then, the optimality of regular play for FP's next move is an obvious consequence of two facts:
    \begin{enumerate}[noitemsep,nolistsep,label={(\arabic*)}]
        \item Playing a move that follows regular play puts FP in a non-losing position. \label{fact1}
        \item Playing a move that violates regular play puts FP in a non-winning position. \label{fact2}
    \end{enumerate}
    Let us show that these two facts are indeed true. Fact \ref{fact1} comes from the observation that FP will always have the option to keep following regular play for as long as SP also does: if SP ever deviates from regular play, then FP can win by Claim \ref{cla:optimalSP}, otherwise we reach the end of regular play and FP has (at least) a drawing strategy at that point by Claim \ref{cla:resultMM}. 
    Fact \ref{fact2} relies on the fact that the move being suggested by regular play is not a greedy move. Indeed, in such situations, Claim \ref{cla:optimalmaker} ensures that, in the Maker-Breaker convention, violating regular play gives Breaker a winning strategy for the Maker-Breaker game. SP could borrow that strategy to (at least) draw the game in the Maker-Maker convention. \qedhere
\end{itemize}	
\end{proof}

Putting Claims \ref{cla:resultMM}, \ref{cla:optimalSP} and \ref{cla:optimalFP} together, we can see that FP has a winning strategy for the Maker-Maker game on $H$ if and only if Alice has a winning strategy for the Geography game on $(N,A,s)$. This concludes the proof of Theorem \ref{the:mainMM}.

\begin{remark}
	We have actually shown that the Maker-Maker problem is {\sf PSPACE}-complete even when restricted to 4-uniform hypergraphs of bounded maximum degree. More specifically, as noted in Remark \ref{rem:boundeddegree}, our construction has maximum degree at most 12, which can be brought down to 5 if sacrificing 4-uniformity to get a hypergraph of rank 4 instead. This improves on the known {\sf PSPACE}-completeness result for the Maker-Maker game on hypergraphs of rank 13 and maximum degree 5 \cite{boundeddegreeQBF}.
\end{remark}

\section{{\PSPACE}-completeness of the vertex-$C_4$-game}\label{section5}\strut
\indent The vertex-$C_4$-game is a specific 4-uniform positional game, where the players take turns picking vertices of a (simple undirected) graph $G$ and the winning sets are the vertex sets of 4-cycles in $G$. We have the following strengthening of Theorems \ref{the:mainMB} and \ref{the:mainMM}.

\begin{theorem}\label{the:C4game}
	Deciding whether the first player has a winning strategy for the vertex-$C_4$-game is a {\sf PSPACE}-complete problem in both Maker-Breaker and Maker-Maker conventions.
\end{theorem}

\begin{proof}
	We just need to show that, denoting by $(N,A,s)$ an instance of the Geography game and by $H$ the hypergraph constructed in the proof of Theorems \ref{the:mainMB} and \ref{the:mainMM}, there exists a graph $G$ on the same vertex set as $H$ such that the vertex sets of 4-cycles in $G$ are the edges of $H$. For all $v \in N$, let $G_v=(V_v,E_v)$ be the graph version of the gadget hypergraph $H_v$, as drawn in Tables \ref{tab:tableau1}--\ref{tab:tableau2}--\ref{tab:tableau3}: it can easily be checked that the vertex sets of 4-cycles in $G_v$ are the edges of $H_v$. What remains to be shown is that, when taking the union $G=\bigcup_{v \in N} G_v$, we do not accidentally create any additional 4-cycle. 
	
	Suppose for a contradiction that there exists a 4-cycle $C=(V_C,E_C)$ in $G$ such that, for all $v \in N$, $C$ is not a subgraph of $G_v$. In particular, let $N_C \subseteq N$ be of minimum size such that $E_C \subseteq \bigcup_{v \in N_C} E_v$: this means that $|N_C| \neq 1$. To conclude, we will use our assumption on the Geography game that there is no cycle (oriented or not) of length less than 6 in the digraph $(N,A)$, as per Theorem \ref{the:geography}. First of all, it is impossible that $|N_C| \in \{3,4\}$: indeed, the nodes in $N_C$ would then form a (not necessarily oriented) 3-cycle or 4-cycle in the digraph $(N,A)$, a contradiction. Therefore, we have $|N_C|=2$. Write $N_C=\{v_1,v_2\}$ such that $a=\overrightarrow{v_1v_2} \in A$. Note that $V_{v_1} \cap V_{v_2}=\{p_a,q_a\}$: this is due to the fact that there is no 2-cycle in $(N,A)$ (indeed, if the reverse arc $\overrightarrow{v_2v_1}$ was also present, then $G_{v_1}$ and $G_{v_2}$ would share two junction pairs instead of one). By minimality of $|N_C|$, we can write $C=x_1x_2x_3x_4$ such that $x_1x_2 \in E_{v_1} \setminus E_{v_2}$ and $x_2x_3 \in E_{v_2} \setminus E_{v_1}$. Moreover, since $p_aq_a \in E_{v_1} \cap E_{v_2}$ (indeed, as apparent in the figures from Tables \ref{tab:tableau1}--\ref{tab:tableau2}--\ref{tab:tableau3}, a junction pair is an edge in every gadget it appears in), the only possibility is that: $x_1 \in V_{v_1} \setminus \{p_a,q_a\}$, $x_2 \in \{p_a,q_a\}$, and $x_3 \in V_{v_2} \setminus \{p_a,q_a\}$. Since $x_4$ is a common neighbor of $x_1$ and $x_3$ in $G_{v_1} \cup G_{v_2}$, this means that $x_4 \in \{p_a,q_a\}$. All in all, $p_a$ and $q_a$ have a common neighbor in $V_{v_1}$ (namely, $x_1$). However, we can see in the figures from Tables \ref{tab:tableau1}--\ref{tab:tableau2}--\ref{tab:tableau3} that the two vertices of a junction pair never have a common neighbor. This is a contradiction.
\end{proof}

\begin{remark}
	We have actually shown that, for both Maker-Breaker and Maker-Maker conventions, deciding whether the first player has a winning strategy for the vertex-$C_4$-game is a {\sf PSPACE}-complete problem even when restricted to graphs of bounded maximum degree. Indeed, it is easily verified that the graph used in our construction has maximum degree 8 (attained by the junction vertices of $H_s$).
\end{remark}

\section{Conclusion}\label{section6}\strut
\indent As our main result, we have shown that deciding the outcome of the Maker-Breaker game or the Maker-Maker game on a 4-uniform hypergraph is a {\PSPACE}-complete problem. For the Maker-Breaker convention, this closes the complexity gap since the game on hypergraphs of rank 3 is solved in polynomial time \cite{Florian_MB3}. For the Maker-Maker convention, this improves on the previous result for hypergraphs of rank 4 \cite{FlorianJonas_MM4}, but the problem remains open for hypergraphs of rank 3.

Our {\sf PSPACE}-hard constructions have bounded maximum degree, but they contain many
pairs of edges which share two or three vertices. Therefore, one could study the complexity of the Maker-Breaker and Maker-Maker games on hypergraphs that are {\em linear}, meaning that any two distinct edges intersect on at most one vertex. It is also of note that the fragment of \textsc{GeneralizedGeography} that we reduce from uses planar digraphs \cite{Geography} and, as such, our method generates $O(1)$-planar 4-uniform hypergraphs. It would be interesting to see if the Maker-Breaker and Maker-Maker problems remain {\sf PSPACE}-complete on planar 4-uniform hypergraphs.

We have also shown that deciding the outcome of the vertex-$C_4$-game is a {\sf PSPACE}-complete problem for both Maker-Breaker and Maker-Maker conventions. This contrasts with the edge-$H$-game, where the smallest known graphs $H$ for which {\sf PSPACE}-completeness is established (in the Maker-Breaker convention) are a graph of order 51 and size 57 and a tree of order 91 \cite{Hgame}. Actually, these numbers can be lowered using our results. Indeed, the authors use reductions from the Maker-Breaker game on 6-uniform hypergraphs, but reducing from the Maker-Breaker game on 4-uniform hypergraphs instead gives a graph of order 35 and size 39 and a tree of order 45. However, this is still far bigger than what we have obtained for the vertex-$H$-game, where the case $H=C_4$ is {\sf PSPACE}-complete. As such, a possible follow-up would be to study the algorithmic complexity of the edge-$H$-game for some graphs $H$ of size 4. The only known results are polynomial-time algorithms for the edge-$K_{1,4}$-game on forests \cite{Hgame} and the edge-$P_5$-game on forests \cite{Arthur} in the Maker-Breaker convention.

Another convention of positional games is called {\em Avoider-Enforcer}, in which Avoider loses if she fills an edge while Enforcer aims at forcing her to fill an edge. A polynomial-time algorithm is known for hypergraphs of rank 2, and for linear hypergraphs of rank 3 in the case where Avoider plays last \cite{AvoiderEnforcer}, while the problem is {\sf PSPACE}-complete for 6-uniform hypergraphs \cite{NacimValentin}. The other cases remain open, and inspiration from the reductions used for the Maker-Breaker convention might be possible, as these two conventions share some important properties from both being {\em weak games} where the two players have complementary goals. It is noteworthy that the 6-uniform construction used to show {\sf PSPACE}-hardness for the Avoider-Enforcer convention \cite{NacimValentin} is very similar to the one that was used for the Maker-Breaker convention \cite{RW21}.

The algorithmic complexity of \textsc{UnorderedQBF} on 3-CNF formulas is also unknown. This problem has been conjectured to be tractable \cite{RW20}, in stunning contrast to its ordered counterpart \textsc{3-QBF} which is {\sf PSPACE}-complete \cite{QBF}. For now, tractability has only been obtained for some subcases \cite{RW20,Florian_MB3}.

Maker-Breaker games are usually seen as asymmetrical. Yet, each player actually aims at filling a transversal of the other player's edges (a {\em transversal} of a set of edges $\mathcal{E}$ is a set of vertices that intersects every element of $\mathcal{E}$). The Maker-Breaker game can therefore be seen as a type of Maker-Maker game where the players have different different winning sets, a variant which has recently been introduced \cite{FlorianJonas_Eurocomb,FlorianJonas_Arxiv}. As such, there is an argument for adding the set of transversals of $E(H)$ to the input of the Maker-Breaker problem along with $H$, so that each player's winning sets are included. Since the number of transversals may be exponential in the size of $H$, it would be interesting to see whether or not the problem remains {\PSPACE}-hard in that case.

\section*{Acknowledgments}\strut
\indent My thanks go Valentin Gledel for suggesting that the Maker-Breaker result could be extended to the Maker-Maker convention. This research was partly supported by the ANR project P-GASE (ANR-21-CE48-0001-01).

\bibliography{biblio}

@article{Finn,
  title={{Solving Maker-Breaker games on 5-uniform hypergraphs is PSPACE-complete}},
  author={Koepke, F. O.},
  journal={Electronic Journal of Combinatorics},
  volume={32},
  number={4},
  year={2025},
  doi="10.37236/13920"
}

@unpublished{FlorianJonas_Arxiv,
      title={{A unified convention for achievement positional games (full version). Preprint.}}, 
      author={F. Galliot and J. S\'enizergues},
      year={2025},
      howpublished="",
      url={https://arxiv.org/abs/2503.18163}, 
}

@unpublished{FlorianJonas_EuroComb,
      title={{A unified convention for achievement positional games (extended abstract). EuroComb'25: 13th European Conference on Combinatorics, Graph Theory and Applications, 25-29 August 2025, Budapest, Hungary. Proceedings to appear}}, 
      author={F. Galliot and J. S\'enizergues},
      year={2025},
      howpublished=""
}

@Article{Schaefer,
  title =	 {On the complexity of some two-person perfect-information games},
  author =	 { T. J. Schaefer},
  doi =		 {10.1016/0022-0000(78)90045-4},
  journal =	 {J. Comput. Syst. Sci.},
  number =   2,
  volume =	 16,
  year =	 1978,
  pages =	 {185-225}
}

@article{Geography,
author = {Lichtenstein, D. and Sipser, M.},
title = {{Go is polynomial-space hard}},
year = {1980},
issue_date = {April 1980},
publisher = {Association for Computing Machinery},
address = {New York, NY, USA},
volume = {27},
number = {2},
doi = {10.1145/322186.322201},
journal = {J. ACM},
pages = {393-401},
}

@Inbook{Gardner_Shannon,
  title     = "The second scientific american book of mathematical puzzles and diversions",
  chapter="{``Recreational topology''}",
  author    = "M. Gardner",
  year      = 1961,
  publisher = "University Of Chicago Press",
  pages={84-87},
  isbn="978-0226282534",
  edition="2nd"
}

@Article{Lehman,
  title =	 {A solution of the {Shannon} switching game},
  author =	 {A. Lehman},
  journal =	 {J. Soc. Indust. Appl. Math.},
  volume =	 12,
  number =  4,
  year =	 1964,
  pages = {687-725},
  doi="10.1137/0112059"
}

@Inbook{Gardner_Hex,
  title     = {Hexaflexagons and other mathematical diversions},
  chapter="{``The game of Hex''}",
  author    = "M. Gardner",
  year      = 1959,
  publisher = "University Of Chicago Press",
  pages={38-40},
  isbn="978-0226282541",
  edition="2nd"
}

@Article{Hales1963,
  title =	 {Regularity and positional games},
  author =	 {A. W. Hales and R. I. Jewett},
  doi =		 {10.1007/978-0-8176-4842-8_23},
  journal =	 {Trans. Amer. Math. Soc.},
  volume =	 106,
  year =	 1963,
  pages =	 {222-229}
}

@unpublished{Florian_MB3,
      title={{Maker-Breaker is solved in polynomial time on hypergraphs of rank 3. Preprint}}, 
      author={F. Galliot and S. Gravier and I. Sivignon},
      year={2022},
      howpublished="",
      url={https://arxiv.org/abs/2209.12819}, 
}

@article{RW21,
author = {Rahman, M. L. and Watson, T.},
title = {{6-uniform Maker-Breaker game is PSPACE-complete}},
year = {2023},
volume = {43},
number = {3},
doi = {10.1007/s00493-023-00026-7},
journal = {Combinatorica},
pages = {595-612}
}

@Article{erdos,
    author =    "P. Erdös and J. L. Selfridge",
    title=      "{On a combinatorial game}",
    journal=    "J. Combin. Theory Ser. A",
    volume=     "14",
    year=       1973,
    DOI=        "10.1016/0097-3165(73)90005-8",
}

@Article{CE78,
  title =	 {Biased positional games},
  author =	 {V. Chv\'atal and P. Erd\H{o}s},
  doi =		 {10.1016/S0167-5060(08)70335-2},
  journal =	 {Ann. Disc. Math.},
  volume =	 2,
  year =	 1978,
  pages =	 {221-229}
}

@Book{Bec08,
  title     = "Combinatorial Games: Tic-Tac-Toe Theory",
  author =   {J. Beck},
  year      = 2008,
  publisher = "Academic Press",
  address   = "Cambridge",
  isbn="978-0-521-46100-9",
  doi="10.1017/CBO9780511735202"
}

@InProceedings{UnorderedQBF1,
author="Ahlroth, L.
and Orponen, P.",
title="Unordered Constraint Satisfaction Games",
booktitle="Math. Found. Comput. Sci. 2012",
year="2012",
publisher="Springer Berlin Heidelberg",
address="Berlin, Heidelberg",
pages="64-75",
isbn="978-3-642-32589-2",
doi="10.1007/978-3-642-32589-2_9"
}

@article{UnorderedQBF2,
author = {Rahman, M. L. and Watson, T.},
title = {{Complexity of unordered CNF games}},
year = {2020},
issue_date = {September 2020},
publisher = {Association for Computing Machinery},
address = {New York, NY, USA},
volume = {12},
number = {3},
doi = {10.1145/3397478},
journal = {ACM Trans. Comput. Theory},
articleno = {18}
}

@inproceedings{QBF,
    title={Word problems requiring exponential time (preliminary report)},
    author={Stockmeyer, L. J. and Meyer, A. R.},
    year = {1973},
    publisher = {Association for Computing Machinery},
    address = {New York, NY, USA},
    doi = {10.1145/800125.804029},
    booktitle = {Proceedings of the Fifth Annual ACM Symposium on Theory of Computing},
    pages = {1-9},
    location = {Austin, Texas, USA},
    series = {STOC '73}
}

@unpublished{Arthur,
    title = {{Solving the $P_5$-game on forests (in French). Unpublished}},
    author = {Duch\^ene, E. and Dumas, A. and Hilaire, M. and Parreau, A.},
    howpublished="",
    year="2024",
    url = "https://jga2024.sciencesconf.org/579456/document"
}

@InProceedings{FlorianJonas_MM4,
  author =	{Galliot, F. and S\'{e}nizergues, J.},
  title =	{{Maker-Maker games of rank 4 are PSPACE-complete}},
  booktitle =	{43rd International Symposium on Theoretical Aspects of Computer Science (STACS 2026)},
  pages =	{40:1-40:17},
  series =	{Leibniz International Proceedings in Informatics (LIPIcs)},
  year =	{2026},
  volume =	{364},
  publisher =	{Schloss Dagstuhl, Leibniz-Zentrum f{\"u}r Informatik},
  address =	{Dagstuhl, Germany},
  doi =		{10.4230/LIPIcs.STACS.2026.40}
}

@unpublished{AvoiderEnforcer,
    title = {{The Avoider-Enforcer game on hypergraphs of rank 3. Preprint}},
    author = {F. Galliot and V. Gledel and A. Parreau},
    howpublished="",
    year="2025",
    url = "https://arxiv.org/abs/2503.21672"
}

@InProceedings{NacimValentin,
  author =	{Gledel, V. and Oijid, N.},
  title =	{{Avoidance games are PSPACE-complete}},
  booktitle =	{40th International Symposium on Theoretical Aspects of Computer Science (STACS 2023)},
  pages =	{34:1-34:19},
  series =	{Leibniz International Proceedings in Informatics (LIPIcs)},
  year =	{2023},
  volume =	{254},
  publisher =	{Schloss Dagstuhl, Leibniz-Zentrum f{\"u}r Informatik},
  address =	{Dagstuhl, Germany},
  doi =		{10.4230/LIPIcs.STACS.2023.34}
}

@InProceedings{RW20,
author="Rahman, M. L.
and Watson, T.",
title="Tractable Unordered 3-{CNF} Games",
booktitle="LATIN 2020: Theoretical Informatics\textup{, proc. Latin American Symposium on Theoretical Informatics}",
series={Lecture Notes in Comput. Sci.},
volume="12118",
year="2020",
publisher="Springer International Publishing",
address="Cham",
pages="360-372",
doi="10.1007/978-3-030-61792-9_29"
}

@Article{Zetters,
    author =    "R. K. Guy and J. L. Selfridge",
    title=      "{Problem S10: solution by T. G. L. Zetters}",
    journal=    "The American Mathematical Monthly",
    volume=     "87",
    year=       1980,
    pages =     {575-576},
    DOI=        "10.2307/2321433",
}

@Article{Byskov,
  title =	 {{Maker-Maker and Maker-Breaker games are PSPACE-complete}},
  author =	 {J. M. Byskov},
  doi =		 {10.7146/brics.v11i14.21839},
  journal =	 {BRICS Report Series},
  volume =	 11,
  number = 14,
  year =	 2004
}

@book{Milos,
  title={{Positional Games}},
  author={Hefetz, D. and Krivelevich, M. and Stojakovi{\'c}, M. and Szab{\'o}, T.},
  volume={44},
  year={2014},
  publisher={Springer},
  doi="10.1007/978-3-0348-0825-5"
}

@article{Hgame,
title = {{Complexity of Maker–Breaker games on edge sets of graphs}},
journal = {Discrete Applied Mathematics},
volume = {361},
pages = {502-522},
year = {2025},
doi = {10.1016/j.dam.2024.11.012},
author = {E. Duch\^ene and V. Gledel and F. {Mc Inerney} and N. Nisse and N. Oijid and A. Parreau and M. Stojakovi{\'c}},
}

@Article{Reisch,
  title =	 {{Hex ist PSPACE-vollst\H{a}ndig (in German)}},
  author =	 {S. Reisch},
  doi =		 {10.1007/BF00288964},
  journal =	 {Acta Informatica},
  volume =	 15,
  pages = {167-191},
  year =	 1981
}

@article{EvenTarjan,
title = "A combinatorial problem which is complete in polynomial space",
author = "S. Even and R. E. Tarjan",
year = "1976",
doi = "10.1145/321978.321989",
volume = "23",
pages = "710-719",
journal = "Journal of the ACM",
number = "4",
}

@article{domination1,
title = {{Maker–Breaker domination game}},
journal = {Discrete Mathematics},
volume = {343},
number = {9},
pages = {111955},
year = {2020},
doi = {10.1016/j.disc.2020.111955},
author = {E. Duch\^ene and V. Gledel and A. Parreau and G. Renault}
}

@article{domination2,
title = {{The Maker–Maker domination game in forests}},
journal = {Discrete Applied Mathematics},
volume = {348},
pages = {6-34},
year = {2024},
doi = {doi.org/10.1016/j.dam.2024.01.023},
author = {E. Duch\^ene and A. Dumas and N. Oijid and A. Parreau and E. R\'emila}
}

@unpublished{vertexHgame,
      title={{$H$-games played on vertex sets of random graphs. Preprint.}}, 
      author={G. Kronenberg and A. Mond and A. Naor},
      year={2019},
      howpublished="",
      url={https://arxiv.org/abs/1901.00351}, 
}

@InProceedings{boundeddegreeQBF,
author="Oijid, N.",
title="{Bounded degree QBF and positional games}",
booktitle="Algorithms and Complexity",
year="2025",
publisher="Springer Nature Switzerland",
address="Cham",
pages="121-135",
doi="10.1007/978-3-031-92935-9_8"
}

\end{document}